\documentclass[journal, onecolumn]{IEEEtran}

\usepackage{cite}
\usepackage{amsmath, amssymb}
\usepackage{algpseudocode}
\usepackage{graphicx}

\ifCLASSOPTIONcompsoc
\usepackage[caption=false,font=normalsize,labelfont=sf,textfont=sf]{subfig}
\else
\usepackage[caption=false,font=footnotesize]{subfig}
\fi

\usepackage{multirow}
\usepackage{booktabs}

\usepackage{mathptmx}       
\usepackage{helvet}         
\usepackage{courier}        

\usepackage{multicol}        
\usepackage[bottom]{footmisc}

\usepackage{verbatim}

\usepackage{hyperref}
\hypersetup{
  colorlinks,
 	bookmarksnumbered=true,
  pdfauthor={Christopher D. Hollander}
}
\usepackage{bookmark}

\newtheorem{thm}{Theorem}
\newtheorem{cor}[thm]{Corollary}
\newtheorem{lem}[thm]{Lemma}
\newtheorem{mdef}{Definition}

\begin{document}

\title{Distributed Consensus Formation Through Unconstrained Gossiping}

\author{Christopher~D~Hollander,~\IEEEmembership{Student Member,~IEEE,}
        and~Annie~S.~Wu%
\thanks{C.D. Hollander and A.S. Wu are with the Department of Electrical Engineering and Computer Science, University of Central Florida, Orlando,
FL, 32816. USA (e-mail: chris.hollander@gmail.com).}}
\maketitle


\begin{abstract}
Gossip algorithms are widely used to solve the distributed consensus problem, but issues can arise when nodes receive multiple signals either at the same time or before they are able to finish processing their current work load. Specifically, a node may assume a new state that represents a linear combination of all received signals; even if such a state makes no sense in the problem domain. As a solution to this problem, we introduce the notion of conflict resolution for gossip algorithms and prove that their application leads to a valid consensus state when the underlying communication network possesses certain properties. We also introduce a methodology based on absorbing Markov chains for analyzing gossip algorithms that make use of these conflict resolution algorithms. This technique allows us to calculate both the probabilities of converging to a specific consensus state and the time that such convergence is expected to take. Finally, we make use of simulation to validate our methodology and explore the temporal behavior of gossip algorithms as the size of the network, the number of states per node, and the network density increase.
\end{abstract}

\begin{IEEEkeywords}
consensus, distributed consensus algorithms, gossip algorithms, networks, markov chains
\end{IEEEkeywords}

\section{Introduction}

The distributed consensus problem asks how every node in a network can adopt the same state value for a given state variable when there is no centralized coordination mechanism \cite{Ren2005, Moallemi2006, Saber2007, Dimakis2010}. Applications of this problem frequently arise in the self-organization, cooperation, and coordination of complex and multi-agent systems. Some of the more commonly studied areas include resource location, formation flight of UAVs, attitude alignment of clusters of satellites, self-organization, automated highway systems, congestion control in communication networks, load balancing, rendezvous in space, distributed sensor fusion, belief propagation, convention emergence, and task allocation \cite{Dimakis2010, Kempe2004, Moallemi2006, Saber2007, Ren2005, Ren2007, Salazar2010}. The distributed consensus problem is also similar to the study of opinion dynamics and the spread of social norms in computational sociology \cite{Hollander2011a, Hollander2011b}.

In this paper, we propose a solution to the distributed consensus problem that makes use of gossip algorithms that allow for nodes to receive multiple simultaneous transmissions, either through parallel synchronized clocks or random chance. This approach is in contrast to the existing research on gossip algorithms in which it is assumed that each node receives only a single incoming transmission \cite{Boyd2006, Dimakis2010}. 

Our solution builds on the representation of gossip algorithms as linear systems \cite{Boyd2006} where the state update equation is defined as $\mathbf{x}(t+1) = \mathbf{W}(t) \mathbf{x}(t)$; where $\mathbf{x}(t)$ is the state vector of the nodes at time $t$ and $\mathbf{W}$ is a random matrix that determines how each node updates its state at time $t$. $W_{ij} = 1$ if node j transmits to node i, and 0 otherwise. Under this definition, the reception of multiple simultaneous transmissions can result in undesired values of $\mathbf{x}(t+1)$ by allowing nodes to take on a value that is a linear combination of the transmitting nodes. In order to avoid this undesired behavior, we redefine the state update equation as $\mathbf{x}(t+1) = f(\mathbf{W}(t)) \mathbf{x}(t)$; where $f(\mathbf{W}(t))$ is a binary relation between $\mathbf{W}(t)$ and a set of row stochastic matrices composed of $\{0,1\}$ entries. For simplicity, we call $f$ the \emph{conflict resolution algorithm}.

The motivation for our research stems from two fronts. First, gossip algorithms are widely used in solving the distributed consensus problem. As computing power increases and computing components become cheaper, it will become easier and easier to build very large decentralized systems. As the size of these systems increases, so too will the complexity of coordinating them. It is therefore important to ensure that they are robust to errors or other unforeseen events such as the reception of multiple simultaneous signals by a single autonomous component.

Second, the theoretical investigation of gossip algorithms is primarily centered on asynchronous and synchronous timing models, both of which guarantee that only two nodes are ever active at once. This ensures that no node ever receives more than a single transmission. This assumption simplifies analysis but, as we show, leaves open the door for interesting and unexpected behaviors upon violation. Such violations may occur as a result of misaligned clocks that allow multiple nodes to fire at once, or as a result of nodes receiving incoming transmissions faster than they are able to process them. If a node processes information slower than it receives it, then the node must either ignore the incoming information, or store it in a queue. If a queue is used, and only one transmission is handled at a time, then there is a risk of queue overflow. If, however, multiple elements are processed from the queue, or if it is emptied and processed as a set of multiple simultaneous transmissions, special treatment will be required to resolve conflicting information. 

Our solution to the distributed consensus problem is designed to handle conflicts that occur as a result of multiple simultaneous transmission by specifying a conflict resolution algorithm for each node. We will show that our solution is guaranteed to produce a consensus when certain assumptions hold and will describe one simple conflict resolution algorithm based on the random selection of an incoming state value. Furthermore, we will show that it is possible to predict the expected consensus state as well as the time required to reach that state. Finally, we will provide empirical data from computer simulations to validate our theoretical claims and then use our theory to explore how various network characteristics impact the temporal behavior of our solution.

A note on notation: we indicate matrices and vectors with bold upper and lowercase symbols, e.g. $\mathbf{M}$ for matrices and $\mathbf{v}$ vectors. Individual elements will be indexed, non-bold, lowercase symbols, e.g. $m_{ij}$ for matrices and $v_{i}$ for vectors.

\section{Related Work}

Given a network $G = (V, E)$ in which every node possesses a state variable, $x$, a gossip algorithm is a method of decentralized information exchange in which one node $u \in V$ transmits the contents of its state variable, $x_u$, to another node, $v \in V$, where $v$ is selected in accordance to some \emph{gossip mechanism}. Upon reception of node $u$'s transmission, node $v$ updates its own state variable, $x_v$, according to some \emph{gossip protocol} \cite{Karp2000, Kempe2002}. Much of the recent work on gossip algorithms make use of a gossip mechanism in which node $v$ is selected uniformly at random from the local neighborhood of node $u$ \cite{Kempe2002, Ganesh2003, Kempe2004, Dimakis2010}, but gossip mechanisms do exist where node $v$ is selected from the entire population \cite{Alon1985, Pittel1987, Demers1988, Hedetniemi1988, Feige1990, Hromkovic1996, Eugster2004, Dimakis2006, Dimakis2008} or node $u$ transmits to all local neighbors via flooding \cite{Kempe2004} or broadcasting \cite{Aysal2009}. Transmission can also be bidirectional \cite{Karp2000, Boyd2006, Cai2010} or unidirectional \cite{Schmalz2007, Fagnani2008, Schmalz2009}. In all cases, the frequency of transmission is controlled by an internal clock that ticks according to either an asynchronous timing model or synchronous timing model.

Under an asynchronous timing model, every node in the network has a clock which ticks according to a Poisson process with a rate of $\lambda$. This is equivalent to a single clock that ticks according to a Poisson process with a rate of $n \lambda$, where $n = |V|$ is the number of nodes in the network \cite{Boyd2006}. In practice, this means that at each tick an average of $n \lambda$ nodes are chosen independently and uniformly at random to transmit their state values.

Under a synchronous timing model, every node in the network has a clock that ticks at the same frequency. This results in all nodes transmitting their state values at exactly the same time. Because transmission occurs in parallel, any information received during a tick cannot be propagated further until the following tick. If it is required that transmission be pairwise disjoint, precautions must be taken to ensure that nodes are not the targets of multiple transmissions when a synchronous timing model is used.

In the existing analysis of gossip algorithms under both synchronous and asynchronous timing models, the \emph{gossip constraint} is often observed. The gossip constraint is responsible for the assumption a node will never receive multiple simultaneous transmissions because with a probability of at least $1 - 1/n$ \cite{Boyd2006} only two nodes will ever be in communication at the same time in the case of bidirectional transmission; and in the case of unidirectional transmission, only one node will transmit at a time. For the synchronous timing model, this implies that nodes must be matched so that their transmissions form disjoint pairs. This constraint greatly simplifies analysis and implementation of gossip algorithms in real world systems \cite{Dimakis2008}, but as we will show through the creation of gossip protocols that handle conflicting information at the node level, it is not strictly required in order to obtain a consensus. Furthermore, by allowing conflict to occur, it becomes much more natural to implement synchronous timing models, since matching is not required, and improbable conflict events under asynchronous timing models require no special treatment or attention. 

When information is spread through gossip in such a way that the gossip constraint is violated, we will use the term \emph{unconstrained gossip}. We will call gossip algorithms that are designed to handle the reception of multiple simultaneous signals \emph{unconstrained gossip algorithms}.

\section{Problem Specification}


Consider a directed graph, $G = (V, E)$ defined by a set of $n$ nodes, $V$, and a set of edges, $E = \{(u, v) : u, v \in V\}$, such that node $u$ points to node $v$. The neighbors of node $u$ are given by $N(u) = \{v : (u, v) \in E \wedge u \neq v\}$. Next, assume that time can be broken in to discrete intervals, where $t$ denotes the current interval. Let each node, $u \in V$, possess a clock that ticks $0$ or $1$ times per interval such that the node acts only when the clock ticks, and a state variable $x_{u} \in S$ whose value is time-dependent and in the set of valid state values, $S$. When the clock of node $u$ ticks, the value of $x_u$ is transmitted along a single outgoing edge to a single neighbor, $v \in N(u)$. Node $v$ updates $x_v$ upon reception of $x_u$ according to the \emph{state update equation} $x_v = h(x_u)$.

If $h(\cdot)$ is linear, then the state update equation $x_v = h(x_u)$ can be vectorized to allow the simultaneous transmission of state values. In this vectorized form, transmission at time $t$ can be written as the linear system
\begin{equation}
\label{eqn:state_update}
\mathbf{x}(t+1) = \mathbf{W}(t) \mathbf{x}(t)
\end{equation}
where $\mathbf{x}(t) \in \mathbb{R}^{n \times 1}$ is the column vector of node states at time $t$, and $\mathbf{W}(t) \in \mathbb{R}^{n \times n}$ is the \emph{transmission matrix} that specifies the source and target nodes of the transmission activity at time $t$. $\mathbf{W}$ is an independent and identically distributed random matrix whose value is determined as\footnote{If $\mathbf{W} = \mathbf{I} - \frac{(e_i - e_j)(e_i - e_j)^T}{2}$ then $\mathbf{W}$ represents the average consensus algorithm. If $\mathbf{W} = \mathbf{I} + e_i (e_j - e_i)^T$ then W represents a directed gossip algorithm.}
\begin{equation*}
  w_{ij} =
  \begin{cases}
  1 & \text{if node } j \text{ transmits to node } i \\
  0 & \text{otherwise} 
  \end{cases}
\end{equation*}

Unlike previous research on gossip-based distributed consensus algorithms \cite{Boyd2006, Dimakis2008, Aysal2009, Schmalz2009}, the transmission matrix we examine is not guaranteed to be row stochastic. This is a direct result of not enforcing the gossip constraint and allowing multiple simultaneous transmissions within the network. As a consequence, nodes may take on undefined state values, and consensus may be impossible to reach. In the rest of this paper we discuss how this consequence can be mitigated and describe a set of algorithms for unrestricted gossiping that can lead to the formation of a consensus despite the positive probability that nodes will receive multiple simultaneous transmissions.



\section{Problem Solution}

Given a communication network where nodes can receive multiple simultaneous transmissions, it is possible to reach a consensus state if there exists an algorithm, $f$, that transforms the transmission matrix, $\mathbf{W}$, into a row stochastic matrix, $\mathbf{A}$. It is sufficient for $f$ to produce a row stochastic matrix because it has been previously proven that if a matrix $\mathbf{A}$ is row stochastic, then $\mathbf{x}(t+1) = \mathbf{A}(t) \mathbf{x}(t)$ will converge to a consensus as $t \rightarrow \infty$ \cite{Boyd2006, Schmalz2007, Schmalz2009}. Furthermore, if $\mathbf{A}$ is row stochastic, then the consensus state is a fixed point of the system and as such will not change unless acted upon by an external influence.

Let $f:\mathbb{R}^{n \times n} \rightarrow \mathbb{R}^{n \times n}$ be an algorithm that transforms $\mathbf{W}$ into $\mathbf{A}$, then equation \ref{eqn:state_update} can be rewritten as
\begin{equation}
\label{eqn:state_update_f}
\mathbf{x}(t+1) = f(\mathbf{W}(t)) \mathbf{x}(t)
\end{equation}

We call $\mathbf{A}(t) = f(\mathbf{W}(t))$ the \emph{adoption matrix} at time $t$ and define it as the row stochastic matrix where
\begin{equation*}
  a_{ij} =
  \begin{cases}
  1 & \text{if node } i \text{ adopts the state of node } j \\
  0 & \text{otherwise} 
  \end{cases}
\end{equation*}
The adoption matrix denotes which transmitting nodes will actually be used to update the states of the receiving nodes.

\subsection{Stability of the consensus state}

The consensus state of (\ref{eqn:state_update_f}) is the vector $\mathbf{x}_c$ such that $x_i = x_j$ for all $x_i, x_j \in \mathbf{x}_c$. Thus $\mathbf{x}_c = k \mathbf{1}$, where $k \in S$.

\begin{lem}
\label{lem:stability}
$\mathbf{x}_c$ is a fixed point of $\mathbf{x}(t+1) = \mathbf{A}(t) \mathbf{x}(t)$.
\end{lem}

\begin{proof}
By construction, $\mathbf{A}(t) = f(\mathbf{W}(t))$ is row stochastic, 
so $\mathbf{A}(t) \mathbf{1} = \mathbf{1}$. 
Thus, $\mathbf{1}$ is an eigenvector of $\mathbf{A}(t)$ with an eigenvalue of $\lambda = 1$. 
Because scalar multiples of eigenvectors are also eigenvectors, $\mathbf{x}_c = k \mathbf{1}$ is an eigenvector of $\mathbf{A}(t)$ with an eigenvalue $\lambda = 1$. 
So $\mathbf{A}(t) \mathbf{x}_c = \mathbf{x}_c$, 
and thus the consensus state, $\mathbf{x}_c$, is a fixed point of $\mathbf{x}(t+1) = \mathbf{A}(t) \mathbf{x}(t)$.
\end{proof}

Thus, by lemma \ref{lem:stability}, if the system reaches a consensus it will remain there until acted upon by external forces.

\subsection{Proof of convergence}

Based on the idea of a \emph{consensus graph} from Schmalz \cite{Schmalz2007, Schmalz2009}, let a \emph{consensus sequence}, $A_c$ be a sequence of adoption matrices, $\{\mathbf{A}(t_1), \mathbf{A}(t_2)... \mathbf{A}(t_{\tau})\}$ with $t_1 < t_2 < \cdots < t_\tau$ such that $\mathbf{x}_c = \mathbf{A}(t_{\tau}) \cdots \mathbf{A}(t_2) \mathbf{A}(t_1) \mathbf{x}(t_1)$.

\begin{lem}
\label{lem:spanning_tree}
If $G$ has a directed spanning tree and \emph{G is not a directed ring network}, then a consensus sequence, $A_c$, exists for the system associated with the communication network, $G$.
\end{lem}

\begin{proof}
If $G$ has a directed spanning tree with a root $\omega \in V$, then a sequence of matrices can be constructed that pass down the state variable of $\omega$ to each child, and then from each child to each grandchild, and so on down the tree until all nodes possess the state value of $\omega$. If, however, $G$ is a ring network then the last node would have no choice but to transmit its state value to the parent node. This action would produce an infinite loop. Thus $G$ must not be a directed ring network.
\end{proof}

\begin{figure}
\centering
\includegraphics[width=80mm]{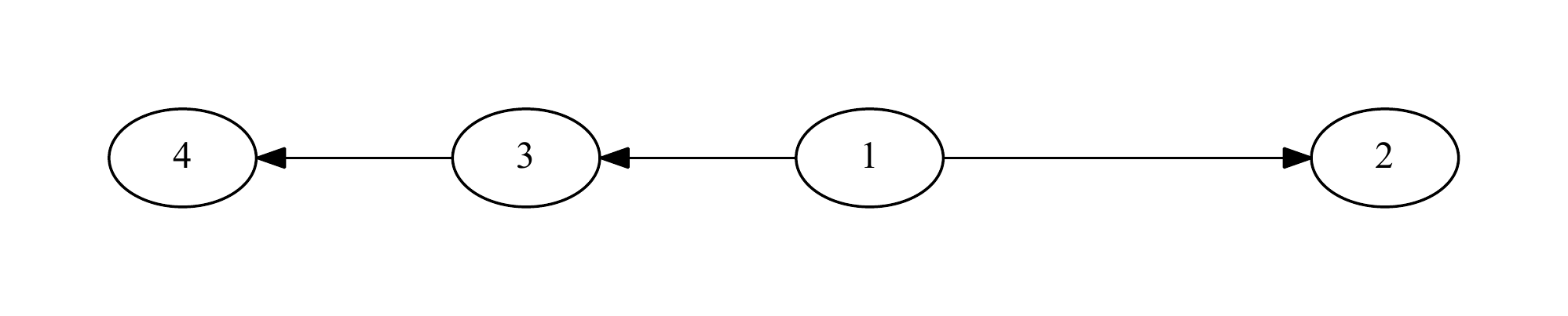}
\caption{The directed graph $G = (\{1, 2, 3, 4\}, \{(1, 2), (1, 3), (3, 4)\})$.}
\label{fig:spanning_tree}
\end{figure} 
Example: Consider the graph $G = (V, E)$ in figure \ref{fig:spanning_tree} where 
\begin{align*}
V & =\{1, 2, 3, 4\} \\
E & = \{(1, 2), (1, 3), (3, 4)\} \\
x(0) & = {[1, 2, 3, 4]}^T \\
\end{align*}

Because this graph contains a directed spanning tree and is not a directed ring network there is a consensus sequence, $A_c = \{\mathbf{A}(t_1), \mathbf{A}(t_2), \mathbf{A}(t_3)\}$, such that 
\begin{align*}
	\mathbf{A}(t_1) & = \begin{bmatrix} 1 & 0 & 0 & 0 \\ 1 & 0 & 0 & 0 \\ 0 & 0 & 1 & 0 \\ 0 & 0 & 0 & 1 \\ \end{bmatrix} \\
	\mathbf{A}(t_2) & = \begin{bmatrix} 1 & 0 & 0 & 0 \\ 0 & 1 & 0 & 0 \\ 1 & 0 & 0 & 0 \\ 0 & 0 & 0 & 1 \\ \end{bmatrix} \\
	\mathbf{A}(t_3) & = \begin{bmatrix} 1 & 0 & 0 & 0 \\ 0 & 1 & 0 & 0 \\ 0 & 0 & 1 & 0 \\ 0 & 0 & 1 & 0 \\ \end{bmatrix} \\
\end{align*}

\begin{thm}
\label{thm:consensus}
If $G$ has a directed spanning tree, is not a directed ring network, and $\mathbf{A}(t) = f(\mathbf{W}(t))$ is row stochastic, then the probability of constructing a consensus sequence, $A_c$, through random selection of transmission matrices tends to 1 as the length of time tends to infinity.
\end{thm}

\begin{proof}
In accordance with previous work on directed gossip algorithms by Schmalz \cite{Schmalz2007, Schmalz2009},
Let $\Delta t$ be a finite time interval.
Let $p$ be the probability that a consensus sequence, $A_c$, occurs within $\Delta t$. Because $G$ has a directed spanning tree, we can set $\Delta t = |A_c|$, and so clearly $p > 0$. 
Let $p_c$ be the event where a consensus sequence, $A_c$, occurs within a time interval $T = r \Delta t$. 
Then,
\begin{align*}
& lim_{r \rightarrow \infty} P(p_c) + lim_{r \rightarrow \infty} P(1 - p_c)  = 1 \\
& lim_{r \rightarrow \infty} P(1 - p_c) = lim_{r \rightarrow \infty} (1 - p)^r   \\
& lim_{r \rightarrow \infty} (1 - p)^r = 0 \text{ since } p > 0                  \\
& lim_{r \rightarrow \infty} P(p_c)                                          = 1
\end{align*}
\end{proof}

\begin{cor}
\label{cor:consensus}
If the communication network, $G$, associated with $\mathbf{x}(t+1) = f(\mathbf{W}(t)) \mathbf{x}(t)$ has a directed spanning tree and is not a ring network, then the system is guaranteed to converge to a consensus, $\mathbf{x}_c$.
\end{cor}

\begin{proof}
By direct application of lemma \ref{lem:spanning_tree} and theorem \ref{thm:consensus}.
\end{proof}

\subsection{An example conflict resolution algorithm: Proportional Selection}
Having shown that the formation of a stable consensus is not only possible, but guaranteed, we now introduce \emph{proportional selection} as an example of a conflict resolution algorithm. Proportional selection takes as input the transmission matrix, $\mathbf{W}$, where $w_{ij} = 1$ if node $j$ transmits to node $i$ and $0$ otherwise. It produces as output the row-stochastic adoption matrix $\mathbf{A}$, such that $a_{ij} = 1$ if node $i$ adopts the state of node $j$ and 0 otherwise, with the additional constraint that $\mathbf{A} \mathbf{1} = \mathbf{1}$.

Under proportional selection, the probability of adoption for each state value in the set of incoming transmissions is equal to the ratio of the number of times that state value occurs to the number of all received state values. For example, if a node receives the values $\{2, 2, 3\}$ then there is a $2/3$ chance of that node adopting the value $2$ and a $1/3$ chance of the node adopting the value $3$.

The proportional selection algorithm consists of two main steps. First initialize the adoption matrix $\mathbf{A} = 0$. Next, for each row, $i$, in $\mathbf{W}$, randomly select a single column, $j$, with a positive entry and then set $A_{ij} = 1$. If a row is composed of all $0$'s then set $A_{ii} = 1$ to denote that the node keeps its current value. This technique guarantees that each row only has a single $1$ and all other entries $0$, thus ensuring $\mathbf{A}$ is row stochastic.

\section{Prediction of the Consensus State and Expected Time to Consensus}

Given a distributed consensus algorithm with the state update equation $\mathbf{x}(t+1) = \mathbf{A}(t) \mathbf{x}(t)$ where $\mathbf{A}(t) = f(\mathbf{W}(t))$ is the adoption matrix at time $t$ and $\mathbf{x}(t)$ is the state vector at time $t$ taken from a finite discrete set of states $S$, it is possible to compute the probability of achieving each consensus state through the construction of a Markov chain; however, to use a Markov chain we must first construct the associated state space, $H$, and transition matrix, $\mathbf{M}$.

\subsection{Generating the Markov State Space}


For the state space of the Markov chain, $H$, we let each element be a vector $\begin{bmatrix}x_0 & x_1 & \cdots & x_n\end{bmatrix}^T$ where $x_i \in S$ is the state variable of the $i$th node in $G$. Under this construction, $H$ is then equal to the set of all possible permutations of $\begin{bmatrix}x_0 & x_1 & \cdots & x_n\end{bmatrix}^T$ as $x_i$ varies over $S$. Given the Markov state space, $H$, we define $\mathbf{z}$ to be a row vector where $0 \leq z_i \leq 1$ denotes the probability that the initial distribution of node values in $G$ is equal to the $i$th Markov state and $\sum_i{z_i} = 1$. 

For example, let $S = \{0, 1\}$, let $h_1 \in H = \begin{bmatrix}0 & 0 & 1\end{bmatrix}^T$ and let $h_5 \in H = \begin{bmatrix}1 & 0 & 1\end{bmatrix}^T$. If $\mathbf{z} = \begin{bmatrix}0 & 1 & 0 & 0 & 0 & 0 & 0 & 0\end{bmatrix}$ then $x_0 = 0$, $x_1 = 0$, and $x_2 = 1$ with probability $1$; however, if $z = \begin{bmatrix}0 & 0.5 & 0 & 0 & 0 & 0.5 & 0 & 0\end{bmatrix}$ then $x_1 = 0$, and $x_2 = 1$ with probability $1$ but $x_0 = 0$ with probability $0.5$ and $x_0 = 1$ with probability $0.5$. This particular definition thus defines $z$ as the starting distribution of the Markov chain. It is anticipated that for most practical applications, the initial distribution of state values will be deterministic, and thus $z_i = 1$ for some $i$; however, if the initial state of each node in $G$ is determined according to a uniform distribution, then $z_i = 1/|H|$ for all $i$. 

The primary application of $\mathbf{z}$ is to explore the distribution of node states in the network at time $t$ through the solution of $\mathbf{z} \mathbf{M}^t$, where $\mathbf{M}$ is the Markov transition matrix.



\subsection{Generating the Markov Transition Matrix}

Given that each state in the Markov chain represents the aggregate state of all nodes in the network, $G$, the transitions between states of the Markov chain represent the change in the distribution of node states. The probabilities of each transition are used to generate the Markov transition matrix, $\mathbf{M}$, where $m_{ij}$ is the probability that the network will transition from state $i$ to state $j$. The specific probabilities depend on not only the topology of the network, but also the conflict resolution algorithm being used. 

For the directed graph $G = (V, E)$, let $G'$ be the adjacency matrix of $G$ with $G'_{ij} = 1$ if node $i$ points to node $j$ and $0$ otherwise. The construction of $\mathbf{M}$ requires multiple stages and depends on $G$, $G'$ and the conflict resolution algorithm.

The first stage in construction of the transition matrix, $\mathbf{M}$, is to generate the state space of the Markov chain, $H$. This is accomplished by enumerating all permutations of the state values that can be taken by the nodes in $G$. For example, on the graph $K_2 = (\{1, 2\}, \{(1, 2), (2, 1)\})$ with $S = \{0, 1\}$, there are $4$ possible Markov states: $\begin{bmatrix}0 & 0\end{bmatrix}^T$, $\begin{bmatrix}0 & 1\end{bmatrix}^T$, $\begin{bmatrix}1 & 0\end{bmatrix}^T$, and $\begin{bmatrix}1 & 1\end{bmatrix}^T$. Every Markov state is represented by a node in the corresponding Markov chain's graph representation.

The second stage is to generate the set of valid transmission matrices, $\mathcal{T}$, for $G$. For every transmission matrix $\mathbf{T} \in \mathcal{T}$, $t_{ij} = 1$ if node $j$ transmits to node $i$ and $0$ otherwise. Because all transmission matrices must account for the graph topology, it must be the case that $t_{ij} = 0$ if $G'_{ij} = 0$. Furthermore, $t_{ii} = 0$ because it is forbidden for a node to transmit to itself. Finally, $\mathbf{T}$ must be column stochastic because we are only considering gossip protocols in which nodes transmit to a single neighbor. 

The third stage is to construct the set of all possible adoption matrices, $\mathcal{A}$. Once all valid transmission matrices are generated, it is guaranteed by construction that they are column stochastic, but not row stochastic. This results in a set of matrices which may represent multiple simultaneous receptions by one or more nodes. In order to resolve this phenomenon, the transmission matrices must be transformed into the row stochastic  adoption matrices through the application of the chosen conflict resolution algorithm. However, because the goal is to construct the Markov transition matrix $\mathbf{M}$, it is essential to construct every possible adoption matrix $\mathbf{A} \in \mathcal{A}$ such that $a_{ij} = 1$ if node $i$ adopts the state of node $j$ and $0$ otherwise; note that in an adoption matrix $a_{ii} = 1$ if node $i$ does not adopt any other node. As a result of this procedure, there is a high probability that duplicate adoption matrices will be generated. These duplicates must be eliminated. 


The fourth stage is to use the set of adoption matrices, $\mathcal{A}$ and the Markov state space, $H$, to generate the edges of the Markov chain's graph representation. For each Markov state $h \in H$, let $h'$ be the set of states reachable from $h$. Then the $i$th transition from $h$ is given by $h'_i = \mathcal{A}_i h$ where $1 \leq i \leq |\mathcal{A}|$. Thus the set of outgoing edges from each Markov state is given as $\{(h, h'_1)$, $(h, h'_2)$, $\cdots$, $(h, h'_{|\mathcal{A}|})\}$ for every $h \in H$. Once all edges have been determined, each of them is assigned a weight\footnote{This specific weight value is due to a uniform probability of transmission among nodes.} of $1/|\mathcal{A}|$. At this point, the Markov chain is represented by a multi-edge digraph. To transform it in to a simple digraph,  first sum the weights of duplicate edges\footnote{We can sum the weights because the probability of transmission independent.} and assign that value to a single edge; then remove all of the duplicates. Once this process is complete, the Markov chain is represented as a simple digraph with weighted edges that determine the transition probability from one state to another. 

The final stage is to specify $\mathbf{M}$. Following the construction of the Markov chain as a simple weighted digraph, $\mathbf{M}$ is then specified as follows. Let $v_i$ and $v_j$ be the $i$th and $j$th node in the Markov graph and $w$ be the weight of the edge $(v_i, v_j)$, then
\begin{equation}
  m_{ij} =
  \begin{cases}
  w & \text{if there is an edge from node } v_i \text{ to node } v_j \\
  0 & \text{otherwise} 
  \end{cases}
\end{equation}
Thus $\mathbf{M}$ is defined as a traditional state transmission matrix from Markov chain theory \cite{Grinstead2006} and each entry $m_{ij}$ represents the probability to transition from state $i$ to state $j$. The exact probability of each entry depends on both the topology of the underlying communication graph and the conflict resolution algorithm.

\subsection{Consequences of Markov Representation}

By representing gossip over a network as a Markov chain, we are able to clearly explain and predict the number of times the nodes of the network will be in a particular configuration, the probability that a particular consensus will be reached, and the expected time that will be required to reach a consensus.

\begin{lem}
\label{lem:abs_chain}
The Markov chain with state space $H$ and transition matrix $\mathbf{M}$ is an absorbing Markov chain.
\end{lem}

\begin{proof}
We have shown, by theorem \ref{thm:consensus}, that given a set of assumptions on G and the use of gossip algorithms described in this paper, a consensus will always be reached if given sufficient time. Furthermore, because the consensus states of G are fixed points, by lemma \ref{lem:stability}, the Markov chain described by $\mathbf{M}$ has one or more absorbing states. Thus, $\mathbf{M}$ is an absorbing Markov chain.
\end{proof}

Because $\mathbf{M}$ is an absorbing Markov chain, it can be rewritten in canonical form \cite{Kemeny1976, Grinstead2006} as
\begin{equation*}
\mathbf{M}' = \begin{bmatrix}\mathbf{Q} & \mathbf{R} \\ 0 & \mathbf{I}\end{bmatrix}
\end{equation*}
where $\mathbf{Q}$ is a submatrix that describes the probability to transition from one transient state to another and $\mathbf{R}$ is a submatrix that describes the probability to transition from a transient state to an absorbing state. This transformation is accomplished by reordering and swapping the rows and columns of $\mathbf{M}$. 

Using $\mathbf{M}'$, we can quickly verify that $\mathbf{M}$ is absorbing by checking to ensure that $\mathbf{Q}^t = 0$ as $t \rightarrow \infty$ \cite{Grinstead2006}. This works because if $\mathbf{M}'$ is absorbing, every transient state should eventually transition to an absorbing state as $t$ tends to infinity. As this occurs, the probability to transition from one transient state to another approaches 0.

As another consequence of $\mathbf{M}'$, the probability to transition from one network configuration, $i$, to another, $j$, after $t$ steps is given by
\begin{equation*}
\mathbf{M}'^t = \begin{bmatrix} \mathbf{Q}^t & \mathbf{R} + \mathbf{RQ} + \mathbf{RQ}^2 + \cdots + \mathbf{RQ}^t \\ 0 & \mathbf{I} \end{bmatrix} 
\end{equation*}

Furthermore, it is now possible to define the \emph{fundamental matrix} \cite{Kemeny1976, Grinstead2006}, $\mathbf{N}$, as $\mathbf{N} = \mathbf{I} + \mathbf{Q} + \mathbf{Q}^2 + \cdots + \mathbf{Q}^t$. If we allow $t$ to approach infinity, then this can be rewritten as $\mathbf{N} = (\mathbf{I} - \mathbf{Q})^{-1}$.

The fundamental matrix is important because it allows us to compute the number of times the nodes of the network will be in a particular configuration, the probability that a particular consensus will be reached, and the expected time that will be required to reach a consensus.




\subsubsection{Calculating the Expected Time in Each State}

\begin{cor}
\label{cor:abs_steps}
$\mathbf{N}$ represents the expected number of times the chain is in state $j$ given that it starts at state $i$.
\end{cor}

\begin{proof}
A direct consequence of Lemma \ref{lem:abs_chain} \cite{Kemeny1976, Grinstead2006}.
\end{proof}

\subsubsection{Calculating the Distribution of Consensus States}

\begin{cor}
\label{cor:abs_prob}
Let $\mathbf{B} = \mathbf{NR}$, then $B_{ij}$ is the probability to be absorbed by the $j$th absorption state given that the initial state of the system is the $i$th transient state.
\end{cor}

\begin{proof}
A direct consequence of Lemma \ref{lem:abs_chain} and Corollary \ref{cor:abs_steps} \cite{Kemeny1976, Grinstead2006}.
\end{proof}

\subsubsection{Calculation of the Expected Time to Consensus}

\begin{cor}
\label{cor:exp_time}
Let $\mathbf{T}_A = \mathbf{N 1}$, then $\mathbf{T}_{A_i}$ is the expected number of steps (or matrix multiplications) until an absorbing state is reached when the system starts in the $i$th transient state.
\end{cor}

\begin{proof}
A direct consequence of Lemma \ref{lem:abs_chain} and Corollary \ref{cor:abs_steps} \cite{Kemeny1976, Grinstead2006}.
\end{proof}

\subsubsection{Rough bounds on the Expected Time to Consensus}

It is also useful to be able to calculate the upper and lower bounds of the expected time to consensus.

Given the expected time to consensus, $\mathbf{t}_A$, the variance of the number of steps is 
\begin{equation*}
\sigma_{t_A}^2 = (2 \mathbf{N} - \mathbf{I}) \mathbf{t}_A - \mathbf{t}_A^2
\end{equation*}
where $\mathbf{t}_A^2$ is the column vector with $\mathbf{t}_{Ai}^2 = \mathbf{t}_{Ai} \mathbf{t}_{Ai}$ \cite{Kemeny1976}. Let X be a random variable representing the time to convergence, then by the Markov inequality $P(X \geq a) \leq \frac{\mathbf{t}_A}{a}$. If $a = k \sqrt{\sigma_{T_A}^2}$ then we can calculate the probability of a value being more than $k$ standard deviations away. If $a = \mathbf{t}_A + \delta$ and then we can calculate the probability that the expected consensus time is larger than some delta of itself. Finally, if $a = \frac{\mathbf{t}_A}{\epsilon}$, then the Markov inequality tells us that $P(X \geq \frac{\mathbf{t}_A}{\epsilon}) \leq \epsilon$. Using this last value for a, we can see that 95\% of the time consensus will be reached in less than $20 \mathbf{t}_A$ steps. Obviously these are very broad bounds and require computation of the Markov matrix, $\mathbf{M}$. We are currently investigating a method by which to bound the consensus time of the systems described in this paper without computing $\mathbf{M}$.


\subsection{A Simple Example}

Having described how to solve for the absorption probabilities, we now provide a concrete example by considering the completely connected 3 node graph, $K_3$, pictured in figure \ref{fig:example_k3} with $S = \{0, 1\}$ and proportional selection.
\begin{figure}
\centering
\includegraphics[width=40mm]{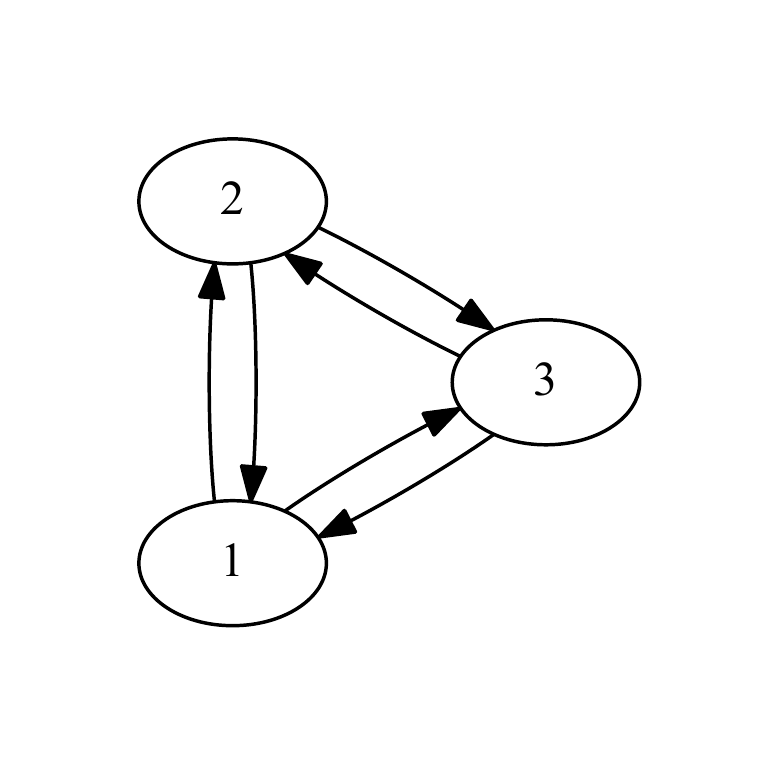}
\caption{The fully connected 3 node graph, $K_3$.}
\label{fig:example_k3}
\end{figure} 

\subsubsection{Step 1: Generate the state space}

There are $|S|^{|K_3|} = 8$ possible Markov states. For simplicity, let us order the Markov states by interpreting each one to be a binary number; thus if $z = \begin{bmatrix}0 & 1 & 0 & 0 & 0 & 0 & 0 & 0 & 0\end{bmatrix}^T$ then the initial state values for each node are given as $x_1 = 0$, $x_2 = 0$, and $x_3 = 1$ and we write the $2$nd Markov state as $\begin{bmatrix}0 & 0 & 1\end{bmatrix}^T$.


\subsubsection{Step 2: Generate the Markov chain}

Figure \ref{fig:example_k3mm} represents the Markov chain corresponding to this example (figure \ref{fig:example_k3}) as a graph with each node representing one possible state configuration for the network, and each edge weight representing the probability of transitioning from one state to another, as determined by the direction of the edge.
\begin{figure}
\centering
\includegraphics[width=80mm]{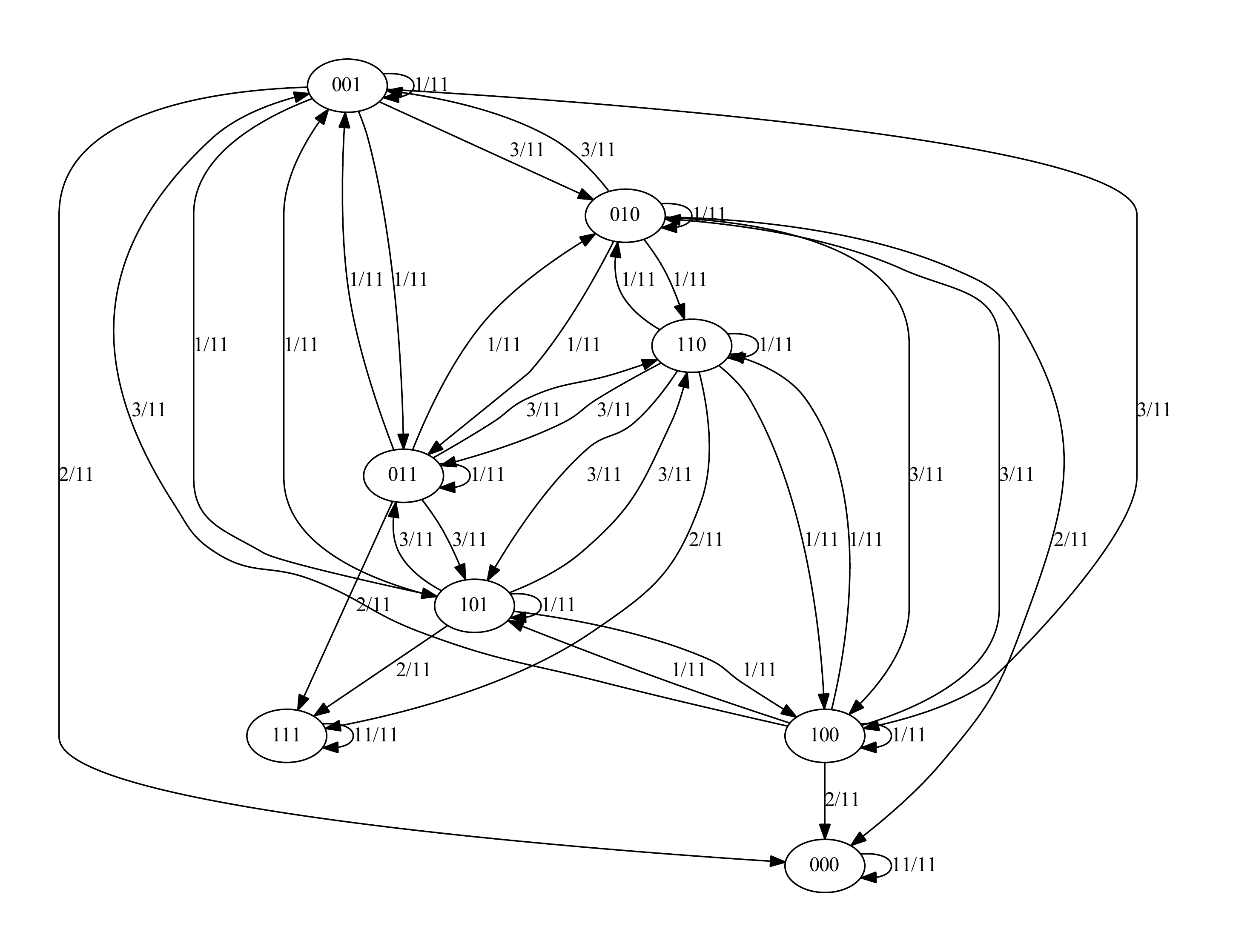}
\caption{A Markov chain for the state value distribution of $K_3$.}
\label{fig:example_k3mm}
\end{figure}
For this particular example it turns out that there are $11$ possible adoption matrices, and so each unaggregated edge is transversed with a probability of $1/11$ due to the use of proportional selection. Removing duplicate edges to consolidate the graph produces the edge weights observed in figure \ref{fig:example_k3mm}.

\subsubsection{Step 4: Generate the transition matrix}

The transition matrix, corresponding to the Markov chain represented in figure \ref{fig:example_k3mm} is given by
\[
\mathbf{M} =
\begin{bmatrix}
1.0 & 0 & 0 & 0 & 0 & 0 & 0 & 0 \\
0.18 & 0.09 & 0.27 & 0.09 & 0.27 & 0.09 & 0 & 0 \\
0.18 & 0.27 & 0.09 & 0.09 & 0.27 & 0 & 0.09 & 0 \\
0 & 0.09 & 0.09 & 0.09 & 0 & 0.27 & 0.27 & 0.18 \\
0.18 & 0.27 & 0.27 & 0 & 0.09 & 0.09 & 0.09 & 0 \\
0 & 0.09 & 0 & 0.27 & 0.09 & 0.09 & 0.27 & 0.18 \\
0 & 0 & 0.09 & 0.27 & 0.09 & 0.27 & 0.09 & 0.18 \\
0 & 0 & 0 & 0 & 0 & 0 & 0 & 1.0 \\
\end{bmatrix}
\]

\subsubsection{Step 5: Solve for the absorption probabilities and expected absorption time}

Placing $\mathbf{M}$ into canonical form, we find that
\[
\mathbf{M}' =
\begin{bmatrix}
0.09 & 0.27 & 0.09 & 0.27 & 0.09 & 0    & 0.18 & 0    \\
0.27 & 0.09 & 0.09 & 0.27 & 0    & 0.09 & 0.18 & 0    \\
0.09 & 0.09 & 0.09 & 0    & 0.27 & 0.27 & 0    & 0.18 \\
0.27 & 0.27 & 0    & 0.09 & 0.09 & 0.09 & 0.18 & 0    \\
0.09 & 0    & 0.27 & 0.09 & 0.09 & 0.27 & 0    & 0.18 \\
0    & 0.09 & 0.27 & 0.09 & 0.27 & 0.09 & 0    & 0.18 \\
0    & 0    & 0    & 0    & 0    & 0    & 1.0  & 0    \\
0    & 0    & 0    & 0    & 0    & 0    & 0    & 1.0  \\
\end{bmatrix}
\]

And so the expected number of steps spend in each state $j$ from the initial state $i$ is given by
\[
\mathbf{N} =
\begin{bmatrix}
    1.7897  &  0.9385  &  0.6329  &  0.9385  &  0.6329  &  0.5675 \\
    0.9385  &  1.7897  &  0.6329  &  0.9385  &  0.5675  &  0.6329 \\
    0.6329  &  0.6329  &  1.7897  &  0.5675  &  0.9385  &  0.9385 \\
    0.9385  &  0.9385  &  0.5675  &  1.7897  &  0.6329  &  0.6329 \\
    0.6329  &  0.5675  &  0.9385  &  0.6329  &  1.7897  &  0.9385 \\
    0.5675  &  0.6329  &  0.9385  &  0.6329  &  0.9385  &  1.7897 \\
\end{bmatrix}
\]

By corollary \ref{cor:abs_prob} the expected probability to reach a consensus on state $j$ given the initial transient state $i$ is given by
\[
\mathbf{B} = \mathbf{N R} = 
\begin{bmatrix}
    0.6667  &  0.3333 \\
    0.6667  &  0.3333 \\
    0.3333  &  0.6667 \\
    0.6667  &  0.3333 \\
    0.3333  &  0.6667 \\
    0.3333  &  0.6667 \\
\end{bmatrix}
\]

By corollary \ref{cor:exp_time} the expected number of steps required to reach a consensus on any of the absorbing states given the initial transient state $i$ is given by
\[
\mathbf{t}_A = \mathbf{N 1} = 
\begin{bmatrix}
    5.5000 \\
    5.5000 \\
    5.5000 \\
    5.5000 \\
    5.5000 \\
    5.5000 \\
\end{bmatrix}
\]

So, based on these results we expect that regardless of the initial distribution of node states it will take on average $5.5$ steps to reach a consensus; but the specific consensus reached \emph{will} depending on the initial distribution of node states.

\section{Theoretical Validation via Numerical Simulation}

It is important that any theoretical framework be validated against empirical, observed, or historical data. We choose to use a simple randomized numerical simulation for this task. We compare our theoretical predictions to empirical data for the consensus probabilities, $\mathbf{B}$, and convergence time, $\mathbf{t}_A$, for a network with a single root node (figure \ref{fig:experiment1}), the completely connected 4 node network, $K_4$ (figure \ref{fig:experiment2}), and the random network in figure \ref{fig:experiment3}.
\begin{figure}
\subfloat[]{\includegraphics[width=40mm]{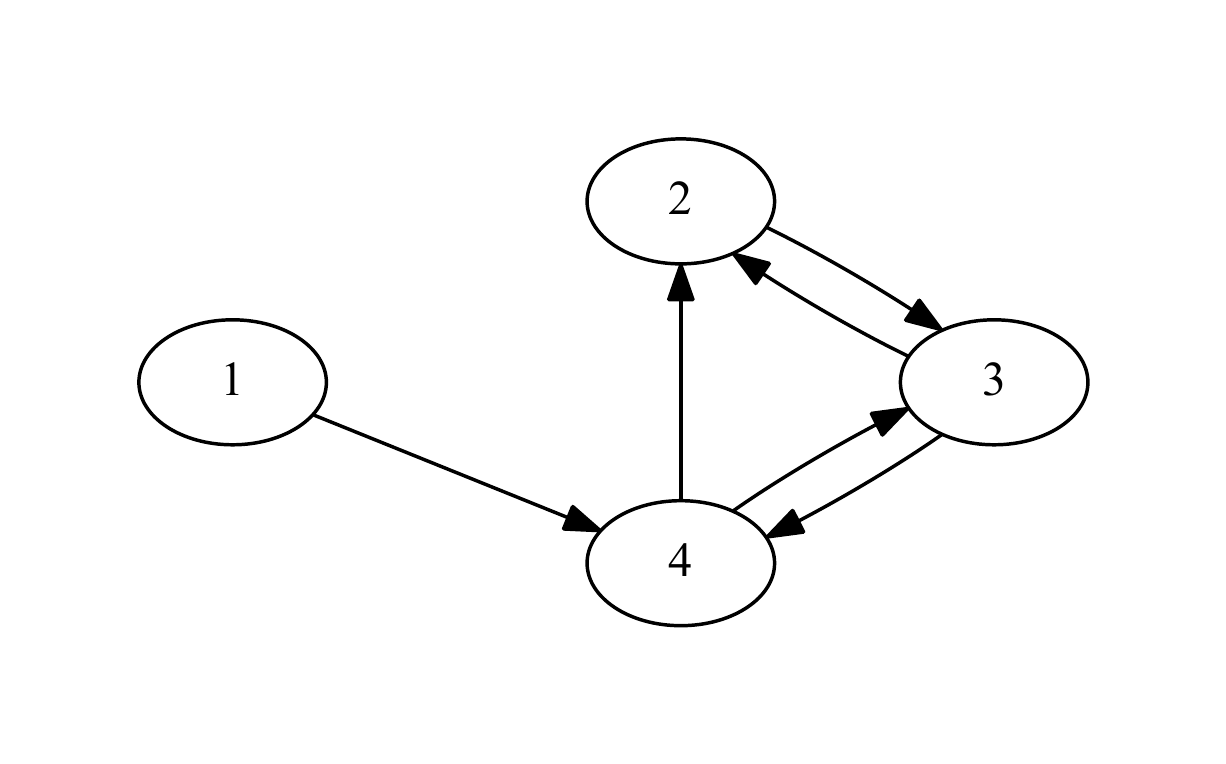} \label{fig:experiment1}}
\hfill
\subfloat[]{\includegraphics[width=40mm]{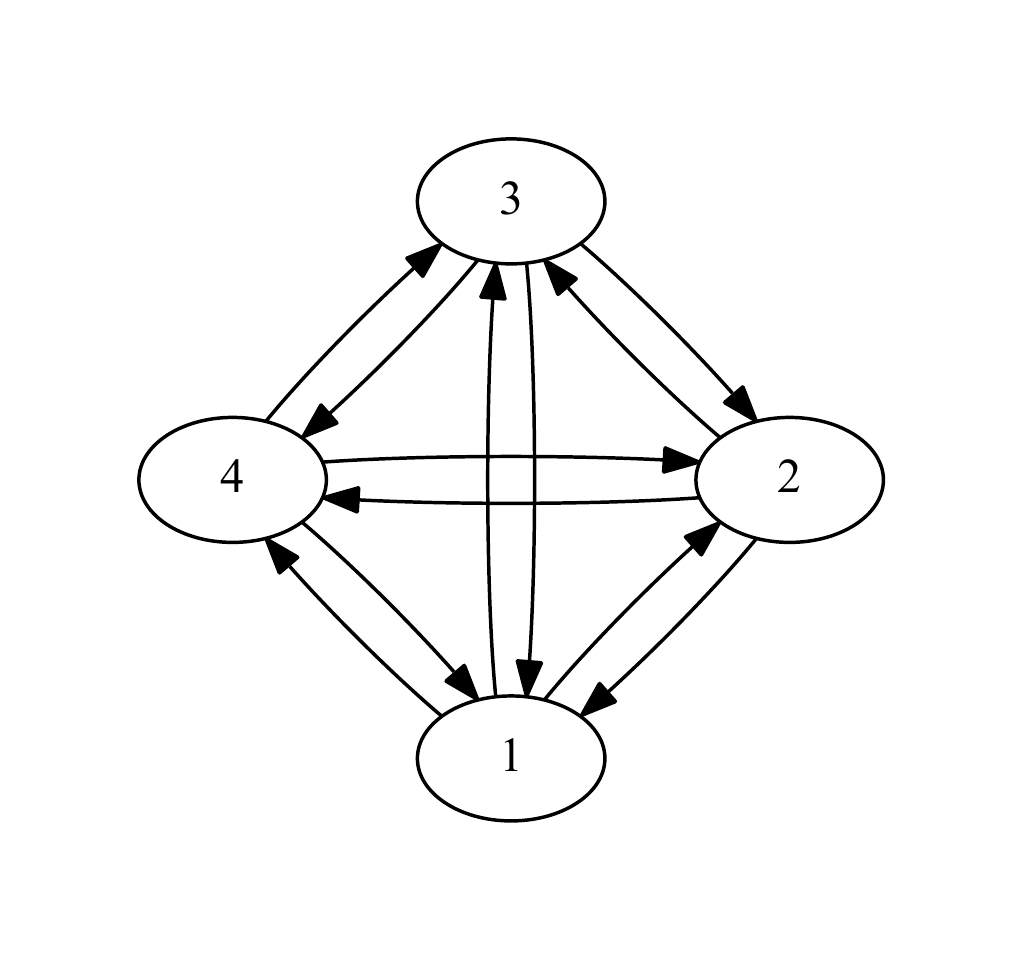} \label{fig:experiment2}}
\hfill
\centerline{\subfloat[]{\includegraphics[width=40mm]{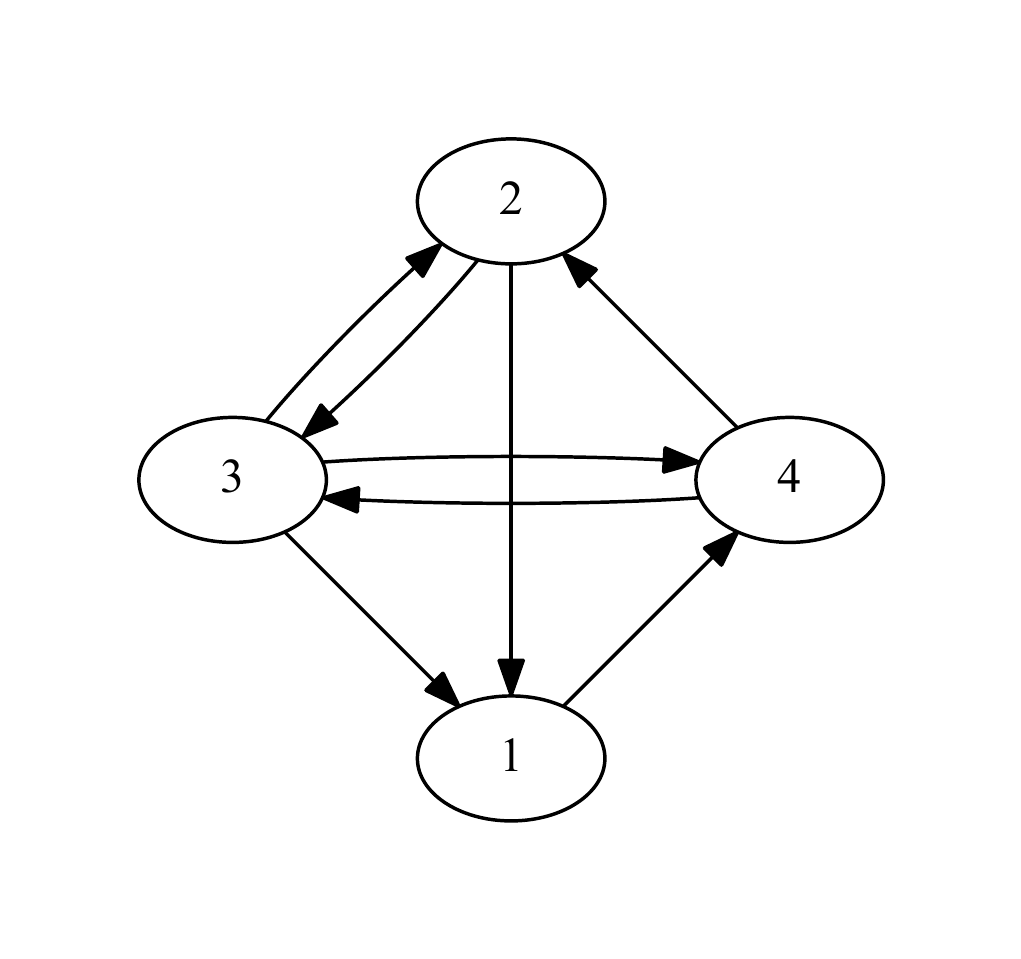} \label{fig:experiment3}}}
\caption{The following networks are used to validate our Markov-based analysis framework: (a) a rooted network, (b) the complete network on 4 nodes, and (c) a random network.}
\label{fig:validation_networks}
\end{figure} 
For each network there two possible states per node and proportional selection is used as the conflict resolution algorithm. We examine the consensus behavior that results from initializing the network with every possible non-consensus network state. Each empirical data point is the mean over 1,000 replications. 

We will claim that our theory is valid if the theoretical predictions and empirical data are approximately equal, such that the empirical data for the consensus probability is within 5\% of the theoretical value and the mean empirical consensus time for all initial states is statistically equal to the corresponding expected theoretical consensus time according to Student's T Test with $\alpha = 0.05$. Due to the nature of pseudo-random numbers we do not expect our theoretical and empirical results to be exactly equal to each other.

To simplify notation, we will favor writing the states of the Markov chain as strings of digits, where the $i$th digit represents the value of the $i$th node, as opposed to using vector notation.

\subsection{Simulation Description}

We represent a system of gossiping agents as the linear dynamical system $\mathbf{x}(t+1) = \mathbf{A}(t) \mathbf{x}(t)$ where $\mathbf{x}$ is the state vector at time $t$ and $\mathbf{A}$ is the adoption matrix at time $t$. The adoption matrix is determined by the conflict resolution algorithm currently being used by the system and the agent's communication network, $G$.

At each time step the simulation executes the following operations. First, the transmission matrix, $\mathbf{T}$, is generated as a uniform random matrix. This construction is done column by column such that there is only a single $1$ in each column.  Next, in preparation for conflict resolution, any row of $\mathbf{T}$ that consists entirely of zeros is replaced by the corresponding unit vector $\mathbf{e}$ such that $e_i = 1$ for the $i$th row and $0$ otherwise. Once $\mathbf{T}$ has been fully generated, it is used to create the value matrix, $\mathbf{T} diag(\mathbf{x}(t))$. This value matrix is then used by the selected conflict resolution algorithm to generate the appropriate adoption matrix, $A$. Upon generation of $A$, the state vector is updated and time is incremented. The simulation comes to a halt when either a user-defined maximum time value is reached, or all values of the state vector $\mathbf{x}(t)$ are epsilon equal, as defined by the equation $max(\mathbf{x}(t)) - min(\mathbf{x}(t)) < \epsilon$.

While the simulation is running we collect the state of every agent at each time step. We also record the halting time. If the vector of states at the final time step are all equal, then the corresponding value is the consensus state of the system. By running a simulation for multiple replications, we are able to calculate the mean halting time as well as the upper and lower bounds for the 95\% confidence interval of the mean halting time. We are also able to use the count of each consensus state to determine the probability of converging to a specific consensus value when starting from a specific initial state. Likewise, we can approximate the convergence probability for a random initial state.

\subsection{Probability to Converge to a Specific Consensus State}

Table \ref{tab:experiment1a} shows the theoretical and empirical probability of consensus over a rooted network (figure \ref{fig:experiment1}), $K_4$ (figure \ref{fig:experiment2}), and a random network (figure \ref{fig:experiment3}). The column header \emph{State} refers to the state of each node in the network such, going left to right, that digit $i$ represents the value of the $i$th node; $\mathbf{x}_t = \mathbf{c}$ is the theoretically expected probability that the network will reach a consensus on state $c$; and $\mathbf{x}_e = \mathbf{c}$ is the empirical probability that the network will reach a consensus on state $c$, as determined by aggregating the simulation data. The \emph{Error} column displays the absolute error between the theoretical and empirical values. 
\begin{table*}
\caption{Theoretical and empirical consensus probabilities for a rooted network, $K_4$, and a random network.}
\label{tab:experiment1a}
\centering
\begin{tabular}{l | r r | r r | r ||  r r | r r | r || r r | r r | r}
\toprule
 & \multicolumn{5}{c||}{(a) Rooted Network (Fig. \ref{fig:experiment1})}
 & \multicolumn{5}{c||}{(b) $K_4$ (Fig. \ref{fig:experiment2})}
 & \multicolumn{5}{c}{(c) Random Network (Fig. \ref{fig:experiment3})} \\
\midrule
 & \multicolumn{2}{c |}{Theoretical}
 & \multicolumn{2}{c |}{Empirical}
 & Error
 & \multicolumn{2}{c |}{Theoretical}
 & \multicolumn{2}{c |}{Empirical}
 & Error
 & \multicolumn{2}{c |}{Theoretical}
 & \multicolumn{2}{c |}{Empirical} 
 & Error \\
\midrule
  State & 
 $\mathbf{x}_{t} = 1$ & $\mathbf{x}_{t} = 2$ & $\mathbf{x}_{e} = 1$ & $\mathbf{x}_{e} = 2$ & Error &
 $\mathbf{x}_t = 1$ & $\mathbf{x}_t = 2$ & $\mathbf{x}_e = 1$ & $\mathbf{x}_e = 2$ & Error &
 $\mathbf{x}_t = 1$ & $\mathbf{x}_t = 2$ & $\mathbf{x}_e = 1$ & $\mathbf{x}_e = 2$ & Error\\
\midrule
1112  &  1.00  &  0.00  &  1.00  &  0.00  &  0.00  &  0.75  &  0.25  &  0.74  &  0.26  &  0.01  &  0.83  &  0.17  &  0.80  &  0.20  &  0.03 \\
1121  &  1.00  &  0.00  &  1.00  &  0.00  &  0.00  &  0.75  &  0.25  &  0.74  &  0.26  &  0.01  &  0.80  &  0.20  &  0.81  &  0.19  &  0.01 \\
1122  &  1.00  &  0.00  &  1.00  &  0.00  &  0.00  &  0.50  &  0.50  &  0.50  &  0.50  &  0.00  &  0.63  &  0.37  &  0.67  &  0.34  &  0.04 \\
1211  &  1.00  &  0.00  &  1.00  &  0.00  &  0.00  &  0.75  &  0.25  &  0.75  &  0.25  &  0.00  &  0.66  &  0.34  &  0.68  &  0.32  &  0.02 \\
1212  &  1.00  &  0.00  &  1.00  &  0.00  &  0.00  &  0.50  &  0.50  &  0.52  &  0.48  &  0.02  &  0.49  &  0.51  &  0.47  &  0.53  &  0.02 \\
1221  &  1.00  &  0.00  &  1.00  &  0.00  &  0.00  &  0.50  &  0.50  &  0.52  &  0.48  &  0.02  &  0.46  &  0.54  &  0.50  &  0.50  &  0.04 \\
1222  &  1.00  &  0.00  &  1.00  &  0.00  &  0.00  &  0.25  &  0.75  &  0.26  &  0.74  &  0.01  &  0.29  &  0.71  &  0.30  &  0.70  &  0.01 \\
2111  &  0.00  &  1.00  &  0.00  &  1.00  &  0.00  &  0.75  &  0.25  &  0.75  &  0.25  &  0.00  &  0.71  &  0.29  &  0.68  &  0.32  &  0.03 \\
2112  &  0.00  &  1.00  &  0.00  &  1.00  &  0.00  &  0.50  &  0.50  &  0.52  &  0.48  &  0.02  &  0.54  &  0.46  &  0.51  &  0.49  &  0.03 \\
2121  &  0.00  &  1.00  &  0.00  &  1.00  &  0.00  &  0.50  &  0.50  &  0.50  &  0.50  &  0.00  &  0.51  &  0.49  &  0.52  &  0.49  &  0.01 \\
2122  &  0.00  &  1.00  &  0.00  &  1.00  &  0.00  &  0.25  &  0.75  &  0.26  &  0.74  &  0.01  &  0.34  &  0.66  &  0.33  &  0.67  &  0.01 \\
2211  &  0.00  &  1.00  &  0.00  &  1.00  &  0.00  &  0.50  &  0.50  &  0.50  &  0.50  &  0.00  &  0.37  &  0.63  &  0.34  &  0.66  &  0.03 \\
2212  &  0.00  &  1.00  &  0.00  &  1.00  &  0.00  &  0.25  &  0.75  &  0.25  &  0.75  &  0.00  &  0.20  &  0.80  &  0.19  &  0.81  &  0.01 \\
2221  &  0.00  &  1.00  &  0.00  &  1.00  &  0.00  &  0.25  &  0.75  &  0.23  &  0.77  &  0.02  &  0.17  &  0.83  &  0.19  &  0.81  &  0.02 \\
\bottomrule
\end{tabular}
\end{table*}

The primary observation to be made from table \ref{tab:experiment1a} is the difference between every theoretical value and its corresponding empirical value is less than $0.05$. Thus, as per our criteria for validity, we claim that our theory correctly estimates the behavior of unconstrained gossip algorithms when conflict resolution is handled via proportional selection.  

Additional observations provide insight into how the topology of a network affects the probability to reach a consensus on a specific state. 
On the rooted network depicted in figure \ref{fig:experiment1}, the state of node $1$ (the left most node) determines the consensus of the system. For example, the network represented by state $1212$ converges to state $1111$; but if the network is initialized according to state $2211$ it will converge to state $2222$. 
On $K_4$ (figure \ref{fig:experiment2}), the probability to reach a consensus on a particular state appears to be related to the ratio of the individual node states in the initial state distribution. For example, the initial state $1121$ has a 75\% chance reach a consensus on state $1$ and a 25\% chance to reach a consensus on state $2$. This behavior appears to be unique to the proportional selection algorithm acting on a completely connected network and most likely arises as a result of the uniform initialization of states, the uniform selection of states during conflict resolution, and the uniform distribution of transmission probabilities due to the completely connected topology of the underlying network.
Data obtained from gossip over a simple random network (figure \ref{fig:experiment3} provides insight into the effects of gossiping over asymmetric and non-rooted networks. These results illustrate that some form of computation or analysis is required to determine the consensus probabilities for all but the most trivial networks.

\subsection{Time Required to Reach a Consensus}

Table \ref{tab:experiment1b} shows the theoretical and empirical values for the average time required to reach a consensus over a rooted network (figure \ref{fig:experiment1}), $K_4$ (figure \ref{fig:experiment2}), and a random network (figure \ref{fig:experiment3}). The column header $E[t]$ is the theoretically expected time till a consensus state is reached by the network; $\mu_{t}$ is the mean time until a consensus is reached as determined by simulation; $95\% CI \mu_{t}$ is the 95\% confidence interval of $\mu_{t}$; \emph{p-value} is the p-value for Student's T Test between the mean time of the empirical data and the expected consensus time. If the p-value is greater than 0.05 then the empirical consensus time is statistically equal to the theoretical consensus time.
\begin{table*}
\caption{Theoretical and empirical average consensus times for a rooted network, $K_4$, and a random network.}
\label{tab:experiment1b}
\centering
\begin{tabular}{l | c | c c c ||  c | c c c || c | c c c}
\toprule
 & \multicolumn{4}{c||}{(a) Rooted Network (Fig. \ref{fig:experiment1})}
 & \multicolumn{4}{c||}{(b) $K_4$ (Fig. \ref{fig:experiment2})}
 & \multicolumn{4}{c}{(c) Random Network (Fig. \ref{fig:experiment3})} \\
\midrule
 & \multicolumn{1}{c |}{Theoretical}
 & \multicolumn{3}{c ||}{Empirical}
 & \multicolumn{1}{c |}{Theoretical}
 & \multicolumn{3}{c ||}{Empirical}
 & \multicolumn{1}{c |}{Theoretical}
 & \multicolumn{3}{c}{Empirical} \\
\midrule
  State & 
 $E[t]$ & $\mu_{t}$ & $95\% CI \mu_{t}$ & p-value &
 $E[t]$ & $\mu_{t}$ & $95\% CI \mu_{t}$ & p-value & 
 $E[t]$ & $\mu_{t}$ & $95\% CI \mu_{t}$ & p-value \\
\midrule
1112  &  3.99  &  2.64  &  (2.40, 2.88)  &  0.00  &  6.13  &  6.66  &  (6.18, 7.14)  &  0.03  &  5.21  &  5.05  &  (4.66, 5.45)  &  0.43 \\
1121  &  3.99  &  2.52  &  (2.30, 2.74)  &  0.00  &  6.13  &  6.74  &  (6.27, 7.20)  &  0.01  &  5.82  &  4.70  &  (4.33, 5.07)  &  0.00 \\
1122  &  6.27  &  4.69  &  (4.44, 4.95)  &  0.00  &  7.73  &  9.34  &  (8.83, 9.85)  &  0.00  &  7.75  &  7.15  &  (6.73, 7.58)  &  0.01 \\
1211  &  5.98  &  3.52  &  (3.27, 3.77)  &  0.00  &  6.13  &  6.74  &  (6.26, 7.21)  &  0.01  &  7.33  &  6.43  &  (6.04, 6.83)  &  0.00 \\
1212  &  7.97  &  5.49  &  (5.23, 5.75)  &  0.00  &  7.73  &  9.01  &  (8.51, 9.51)  &  0.00  &  8.31  &  7.62  &  (7.21, 8.03)  &  0.00 \\
1221  &  7.96  &  5.65  &  (5.39, 5.90)  &  0.00  &  7.73  &  8.55  &  (8.10, 9.00)  &  0.00  &  8.13  &  7.56  &  (7.13, 7.99)  &  0.01 \\
1222  &  9.56  &  6.58  &  (6.33, 6.83)  &  0.00  &  6.13  &  7.02  &  (6.52, 7.52)  &  0.00  &  7.14  &  6.84  &  (6.42, 7.25)  &  0.15 \\
2111  &  9.56  &  6.56  &  (6.32, 6.81)  &  0.00  &  6.13  &  6.86  &  (6.36, 7.37)  &  0.00  &  7.14  &  6.87  &  (6.45, 7.29)  &  0.21 \\
2112  &  7.96  &  5.57  &  (5.30, 5.84)  &  0.00  &  7.73  &  8.88  &  (8.40, 9.36)  &  0.00  &  8.13  &  7.13  &  (6.73, 7.53)  &  0.00 \\
2121  &  7.97  &  5.33  &  (5.07, 5.59)  &  0.00  &  7.73  &  8.91  &  (8.39, 9.43)  &  0.00  &  8.31  &  7.91  &  (7.44, 8.37)  &  0.09 \\
2122  &  5.98  &  3.61  &  (3.37, 3.85)  &  0.00  &  6.13  &  7.32  &  (6.77, 7.87)  &  0.00  &  7.33  &  6.64  &  (6.24, 7.05)  &  0.00 \\
2211  &  6.27  &  4.51  &  (4.26, 4.76)  &  0.00  &  7.73  &  9.18  &  (8.66, 9.69)  &  0.00  &  7.75  &  7.29  &  (6.87, 7.70)  &  0.03 \\
2212  &  3.99  &  2.83  &  (2.59, 3.08)  &  0.00  &  6.13  &  7.22  &  (6.73, 7.72)  &  0.00  &  5.82  &  4.79  &  (4.37, 5.20)  &  0.00 \\
2221  &  3.99  &  2.58  &  (2.36, 2.80)  &  0.00  &  6.13  &  6.65  &  (6.19, 7.11)  &  0.03  &  5.21  &  5.17  &  (4.77, 5.58)  &  0.85 \\
\bottomrule
\end{tabular}
\end{table*}

The primary observation to be made from table \ref{tab:experiment1b} is that there is inconsistency in the statistical equality between the empirical consensus time and the theoretical consensus time. In the case of the rooted network, the empirical time is statistically less than the theoretical time; indicating that our simulation converges to a consensus faster than expected value. In the case of the $K_4$ network, the empirical time is statistically greater than the theoretical time; indicating that our simulation converges slower than the expected value. Finally, in the case of the random network, the empirical and theoretical times are statistically equal for some initial states and not for others. These results appear to be caused by our use of pseudorandom numbers in the simulation. Because the networks tested are small (only 4 nodes), the distribution of node states is not very ``random''. As the number of nodes in the network increases, however, the mean empirical consensus time approaches the theoretical expected consensus time and the p-values increase in response. 
\begin{table}
\caption{Growth in the statistical equality of consensus times over complete networks with proportional selection. }
\label{tab:validation1c}
\begin{center}
\begin{tabular}{l l l}
\toprule
4 nodes        & 5 nodes          & 6 nodes \\
\midrule
(0.00, 5.00)\% & (56.54, 90.13)\% & (82.75, 97.89)\% \\   
\bottomrule
\end{tabular}
\end{center}
\end{table}	
An example of this behavior can be observed in table \ref{tab:validation1c}, where each column represents the 95\% confidence interval for the mean percentage of initial states that converge to a consensus with a mean empirical speed that is statistically equal to the expected theoretical consensus time. These values are based on 1000 replications of the simulation. Based on our findings, we assert that our theory is valid with the footnote that expected times should only be taken as a rough approximation in the case of very small networks.



Further observation of the data in table \ref{tab:experiment1b} reveals that across all networks examined, multiple states within each sample network require the same amount of time to reach a consensus. For example, under the rooted network depicted in figure \ref{fig:experiment1} the states 1112, 1121, 2212, and 2221 all have the same expected consensus time. This observation suggests that the initial state of a network can have a serious impact on the time required to reach a consensus and begs the question, ``what do these initial states have in common?'' We introduce a notion of \emph{distance} between initial states and consensus in order to quantify the commonalities between initial states and provide one answer this question.

\subsection{Distance to Consensus}


We have shown that when information is exchanged over a network through unconstrained gossip it is possible that multiple initial states may result in the same level of performance. To explain why this might be the case, we postulate that states with similar consensus times are also of a similar distance to consensus. In this context, ``distance to consensus'' refers to the distance from a specific initial network state to \emph{any} consensus state. 

\begin{mdef}
Given the network $G = (V, E)$, let $h \in H$ be a distribution of node states in the network and $c$ be a specific consensus state of the network. Both $h$ and $c$ are vectors with the $i$th element representing the state of the $i$th node in the network. We define the distance between $h$ and $c$ as 
\[
D_H(h, c) = \sum\limits_{i=1}^{|V|}{\delta^{-1}(h_i, c_i)} \text{ where }
\delta^{-1}(h, c) = 
\begin{cases}
1, & \text{if }h \neq c \\
0, & \text{if }h = c
\end{cases}
\]
\end{mdef}

It is no coincidence that $\delta_H$ is essentially the hamming distance. 

\begin{mdef}
Given the network $G = (V, E)$, let $h \in H$ be the initial distribution of node states in the network and $C$ be the set of all possible consensus states; e.g. $C_1 \in C = \begin{bmatrix}1 & 1 & \cdots & 1 \end{bmatrix}$. Furthermore, let $P(c | h)$ be the probability that the network reaches a consensus on state $c$ given that it initialized in state $h$. We define the expected distance between $h$ and $C$ to be

\[
D(h, C) = \sum\limits_{c \in C}{P(c | h) * \delta_H(h, c)}
\]

\end{mdef}

The expected distance between $h$ and $C$ is the ``distance to consensus'' from $h$.

It is important to note that our definition of distance accounts for neither the number of possible node states nor the number of nodes in the network. As a consequence, we are able to compare distances across networks of varying topologies. For example, consider the two completely connected networks $K_4$ and $K_5$ that each appear to be one step a particular consensus state. Is the distance from state $h = \begin{bmatrix} 1 & 1 & 2 & 1\end{bmatrix}$ to consensus equal to the distance from $h' = \begin{bmatrix} 1 & 1 & 1 & 2 & 1\end{bmatrix}$? It turns out that the answer is no, they are not exactly equal. The state $h$ is 1.5 units from consensus, while the state $h'$ is 1.6 units because while both states are a single step from a consensus on state $1$, $K_4$ is three steps from a consensus on state $2$ while $K_5$ is 4 steps from a consensus on state 2. This result poses yet another question, ``when are two states close to one another in terms of the distance to consensus?'' To answer this question, we offer up the following definition.

\begin{mdef}
Two states, $h_1$ and $h_2$, are ``close'' to one another in terms of their distance to consensus if $round(h_1) = round(h_2)$; where $round(.) = $ is the elementary rounding function that rounds a real number to the nearest integer. 
\end{mdef}

This notion of two states being close to one another in terms of their distance to consensus is especially useful in the empirical investigation of various network properties on the time required to reach consensus. As we shall see, because the consensus time differs depending on the initial state, and the initial states can be partitioned based on their distance from consensus, we are able to reduce error by conducting analysis on the partitions of states as opposed to aggregating effects over all initial states. 

Furthermore, the ability to classify performance based on the distance to consensus implies that if the performance of one initial state in a class is known, the performance of all other states in that class will be similar. The exact computational requirements needed to classify each state will differ depending on the underlying network and the conflict resolution algorithm, but in the case of certain configurations, such as a completely connected network with proportional selection, it is trivial to determine if two states are in the same class\footnote{In the case of a completely connected network with proportional selection, the probability to converge to a particular consensus state is proportionally determined by the initial state (see table \ref{tab:experiment1a}). $D_H$ is always straightforward to compute.}.

\section{Influences on Consensus Time: The Effect of States, Nodes, and Density}

When applying absorbing Markov chains to analyze the behavior of unconstrained gossip algorithms, it becomes clear that the number of nodes in the network, the number of states each node can represent, the topology of the network, and the conflict resolution algorithm should all be determining factors in the length of time required to reach consensus. The number of nodes and and node states determine the size of the Markov chain. The network topology and conflict resolution algorithm determine the transition probabilities. 

Now that we have shown our analytical framework to be valid, we briefly explore the impact of the Markov chain state space size and communication network density on the consensus time of small networks operating under unconstrained gossip with proportional selection.

\subsection{Expectations}

It is our expectation that under proportional selection the time required to reach consensus will increase as the the number of nodes in the communication network and the number of possible states that each node can assume increases - regardless of the underlying topology. This expectation is based on the growth of the state space for the associated Markov chain. 

We also expect that, in general, the time required to reach a consensus will not necessarily decrease as the density of the underlying network increases. This expectation is due largely to the random nature of gossiping. Dense networks impose a larger number of choices for each node to make on average. In the worst case, this can result in an increase in the consensus time due to poor choice of transmission paths. In the best case, the optimal consensus sequence can be selected and result in the shortest consensus time.

\subsection{The Impact of Nodes and Node States on Consensus Time}

Given a network of $n$ nodes with $k$ possible states per node, each network state in the Markov chain can be represented as a vector of $n$ elements with each element taking on a value between $1$ and $k$, inclusive. Under this representation adding one more node to the network is equivalent to increasing the total number of states in the Markov chain from $k^n$ to $k^{n+1}$, and adding one more possible node state is equivalent to increasing the total number of states in the Markov chain from $k^n$ to $(k+1)^{n}$. Thus, it is reasonable to expect that as the number of nodes in a network and the number of possible states per node grow larger, adding one more node will produce many more Markov states than adding one more possible node state. This increase in the size of the Markov chain should have a direct impact on the time required to reach a consensus, with larger chains requiring more time as a result of the possible states that a network can end up in. Additionally, because the initial states of the network can be partitioned by their distance from consensus, we expect that the impact of adding nodes and node states will be more prevalent when the initial state of the network is far from consensus.

Examples of the consensus time behavior under bidirectional completely connected (figures \ref{fig:k3}, \ref{fig:k4}, and \ref{fig:k5}), star (figures \ref{fig:star3}, \ref{fig:star4}, and \ref{fig:star5}), and ring (figures \ref{fig:ring3}, \ref{fig:ring4}, and \ref{fig:ring5}) networks are displayed in tables \ref{tab:hyp1_1}, \ref{tab:hyp1_2}, and \ref{tab:hyp1_3}. These network structures were chosen for exploration over random and complex networks because they're easier to uniformly scale.
\begin{figure}
\subfloat[]{\includegraphics[width=25mm]{fig/k3.pdf} \label{fig:k3}}
\hfill
\subfloat[]{\includegraphics[width=25mm]{fig/k4.pdf} \label{fig:k4}}
\hfill
\subfloat[]{\includegraphics[width=25mm]{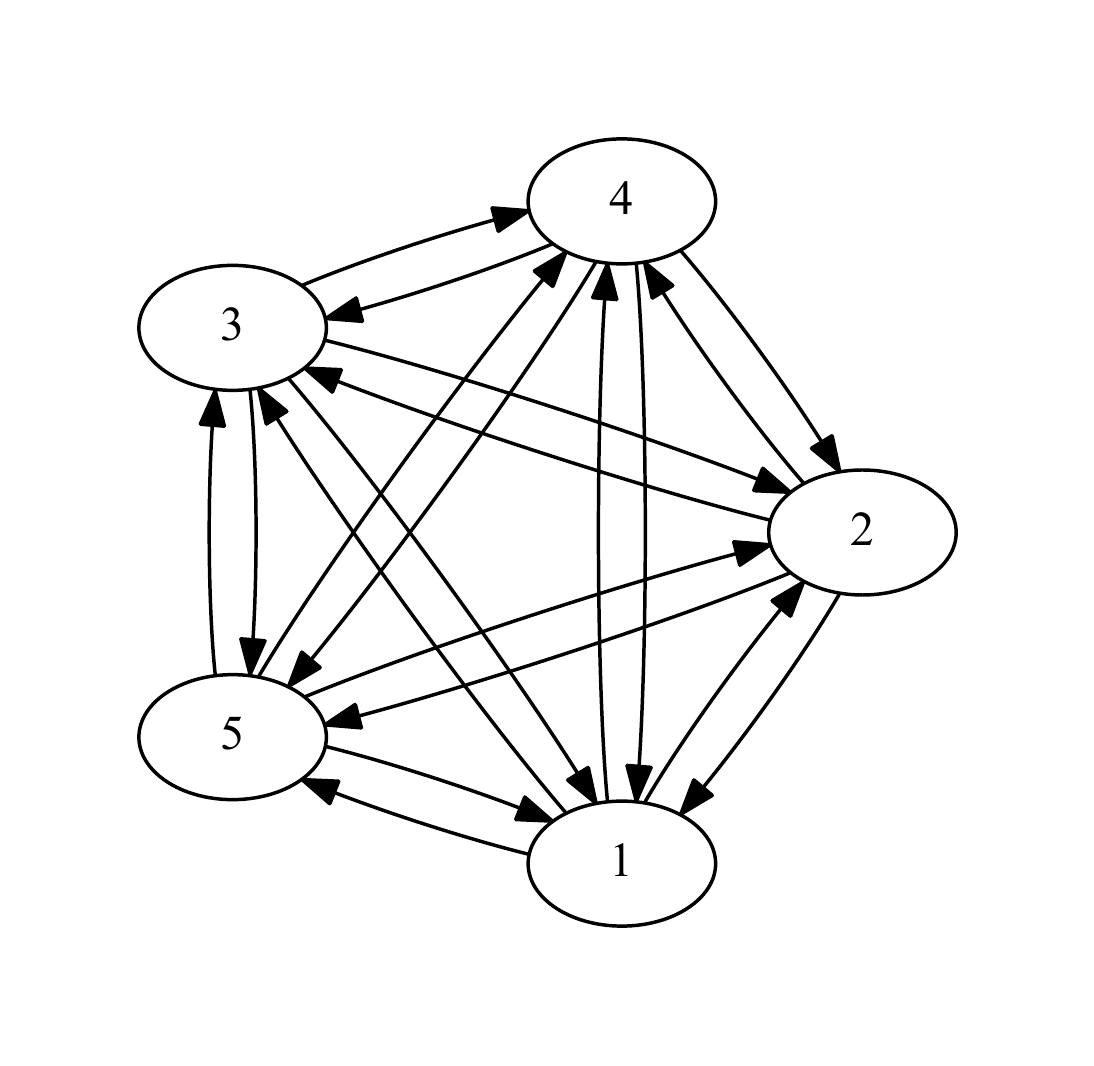} \label{fig:k5}}
\vfill
\subfloat[]{\includegraphics[width=25mm]{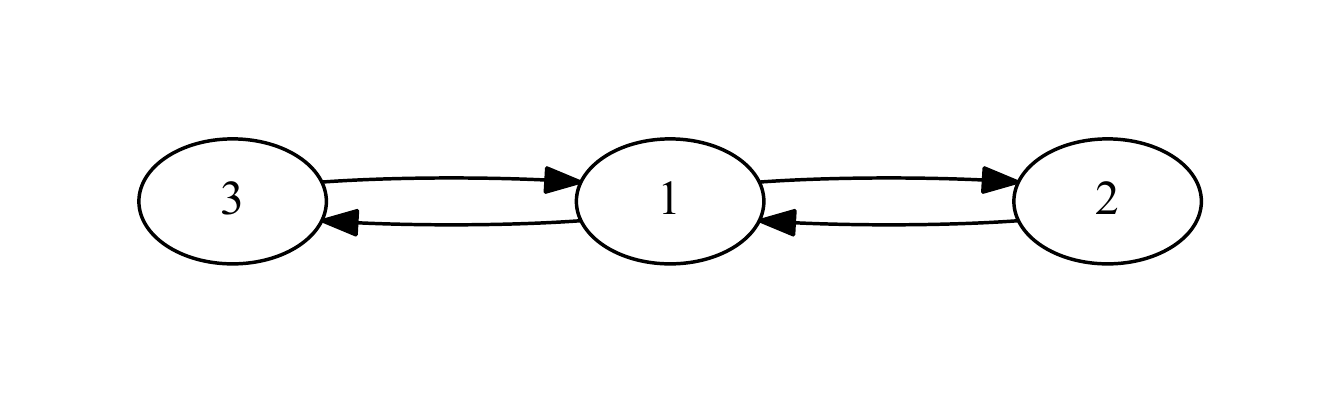} \label{fig:star3}}
\hfill
\subfloat[]{\includegraphics[width=25mm]{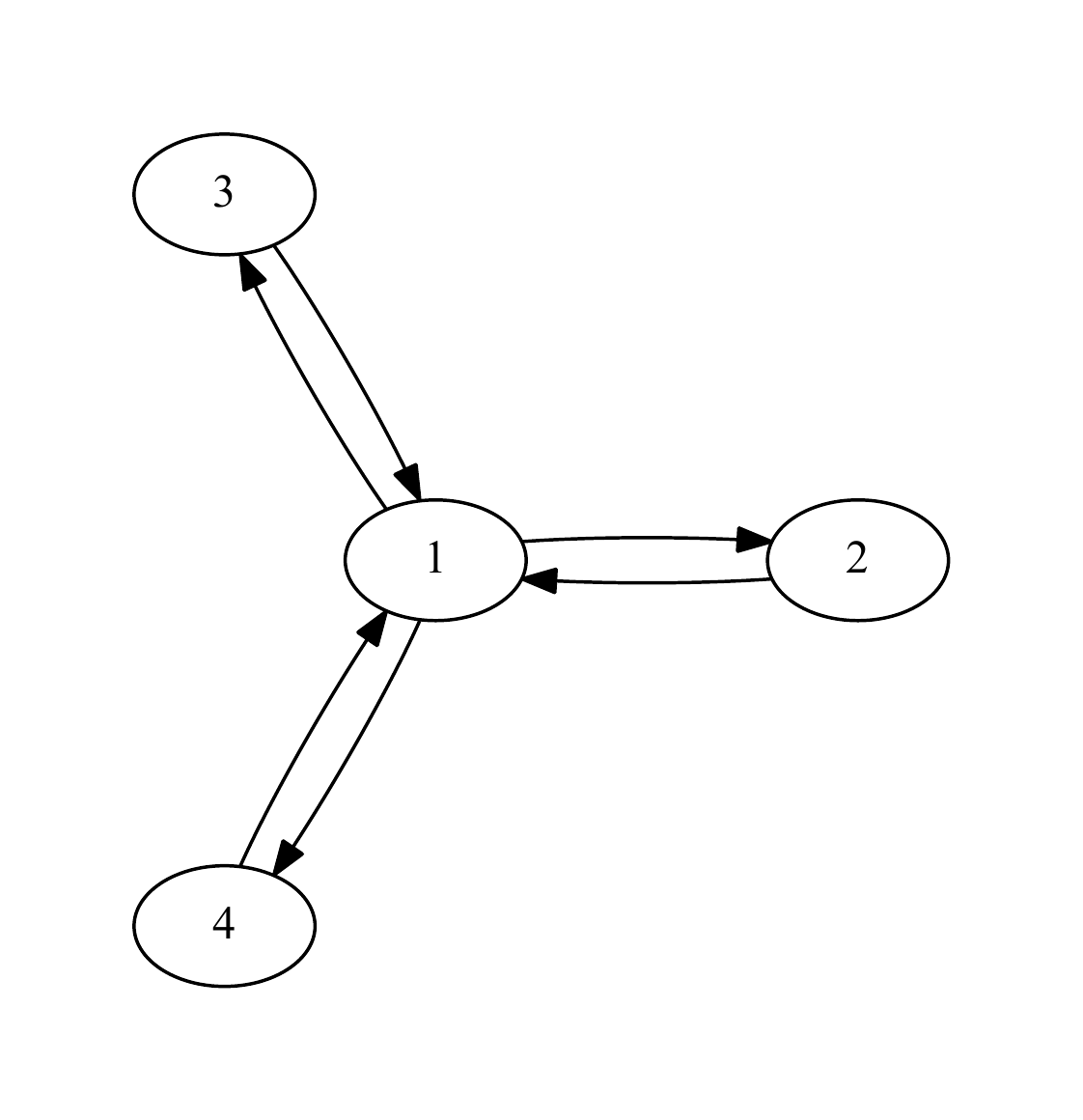} \label{fig:star4}}
\hfill
\subfloat[]{\includegraphics[width=25mm]{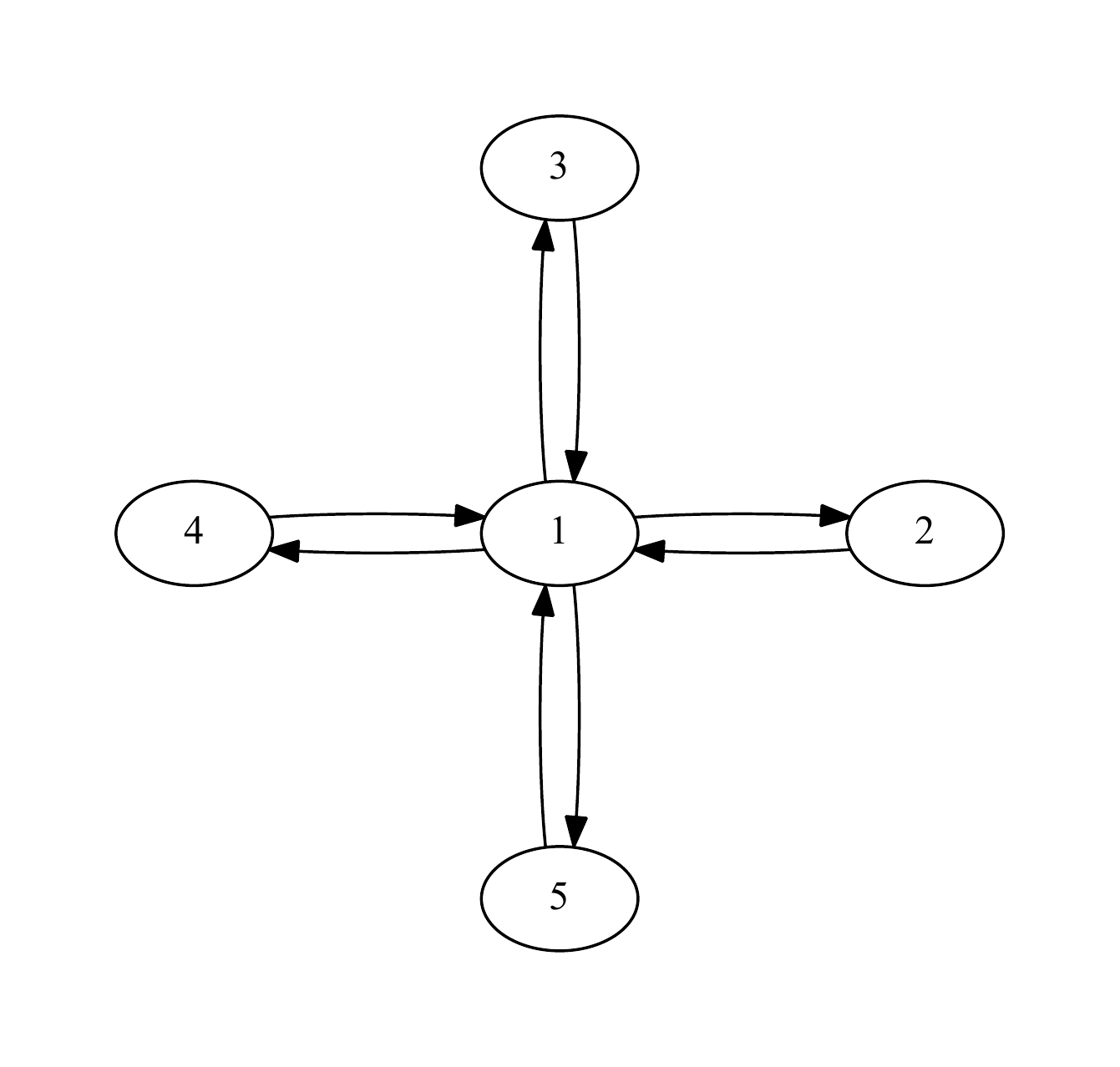} \label{fig:star5}}
\vfill
\subfloat[]{\includegraphics[width=25mm]{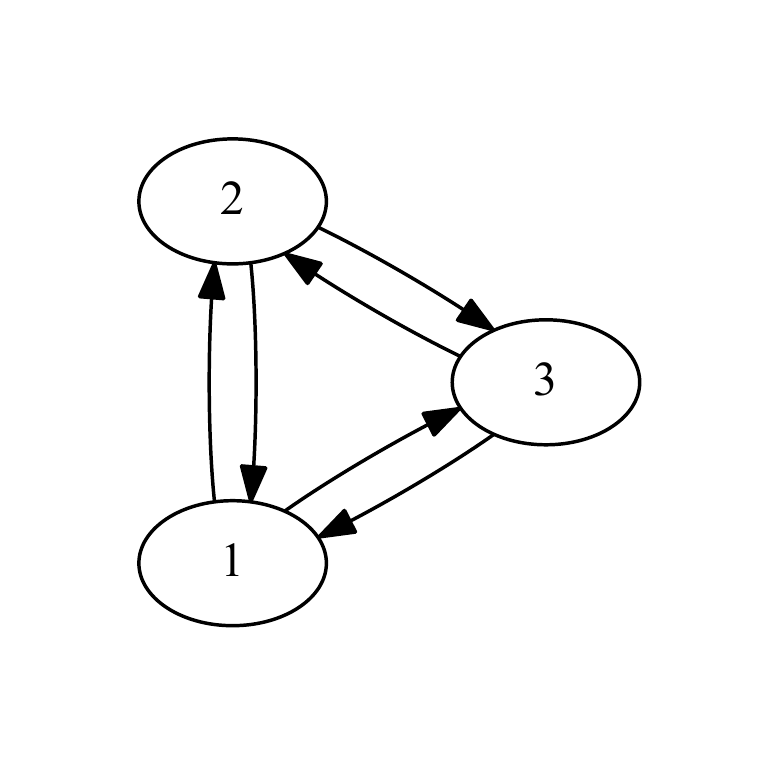} \label{fig:ring3}}
\hfill
\subfloat[]{\includegraphics[width=25mm]{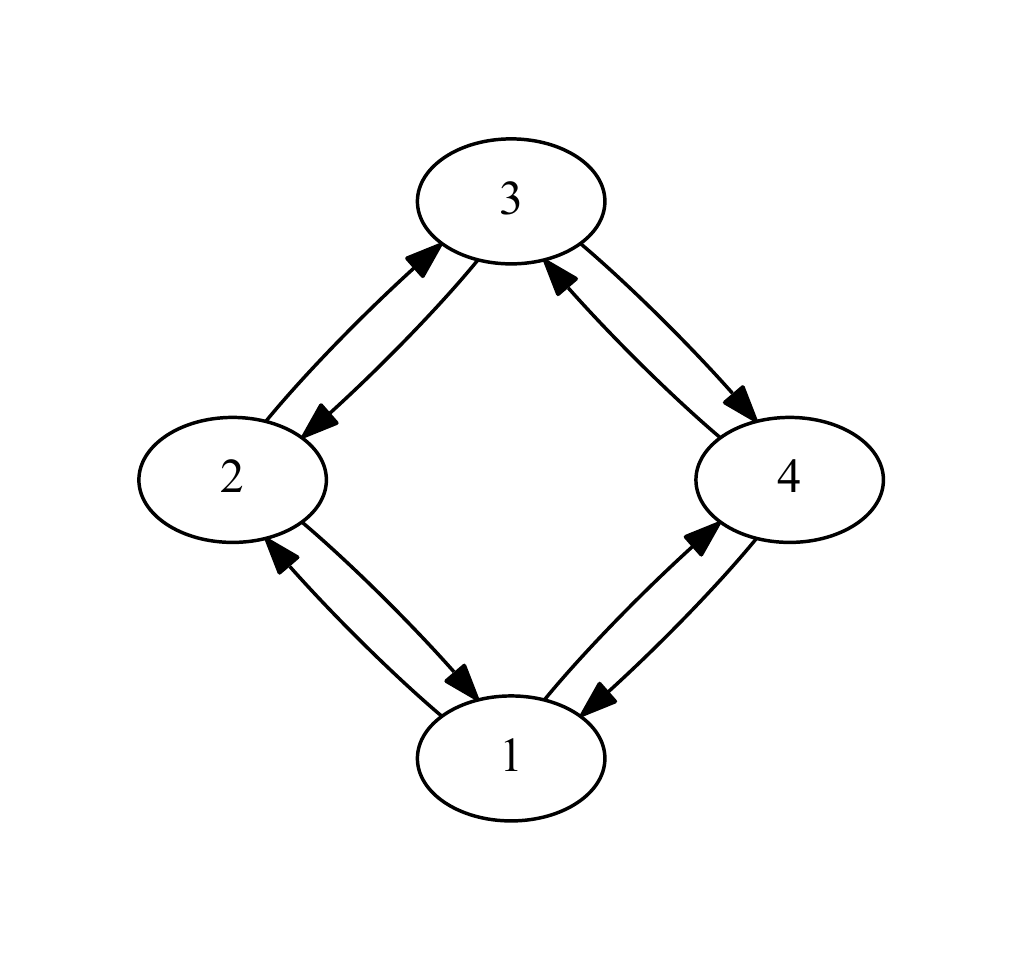} \label{fig:ring4}}
\hfill
\subfloat[]{\includegraphics[width=25mm]{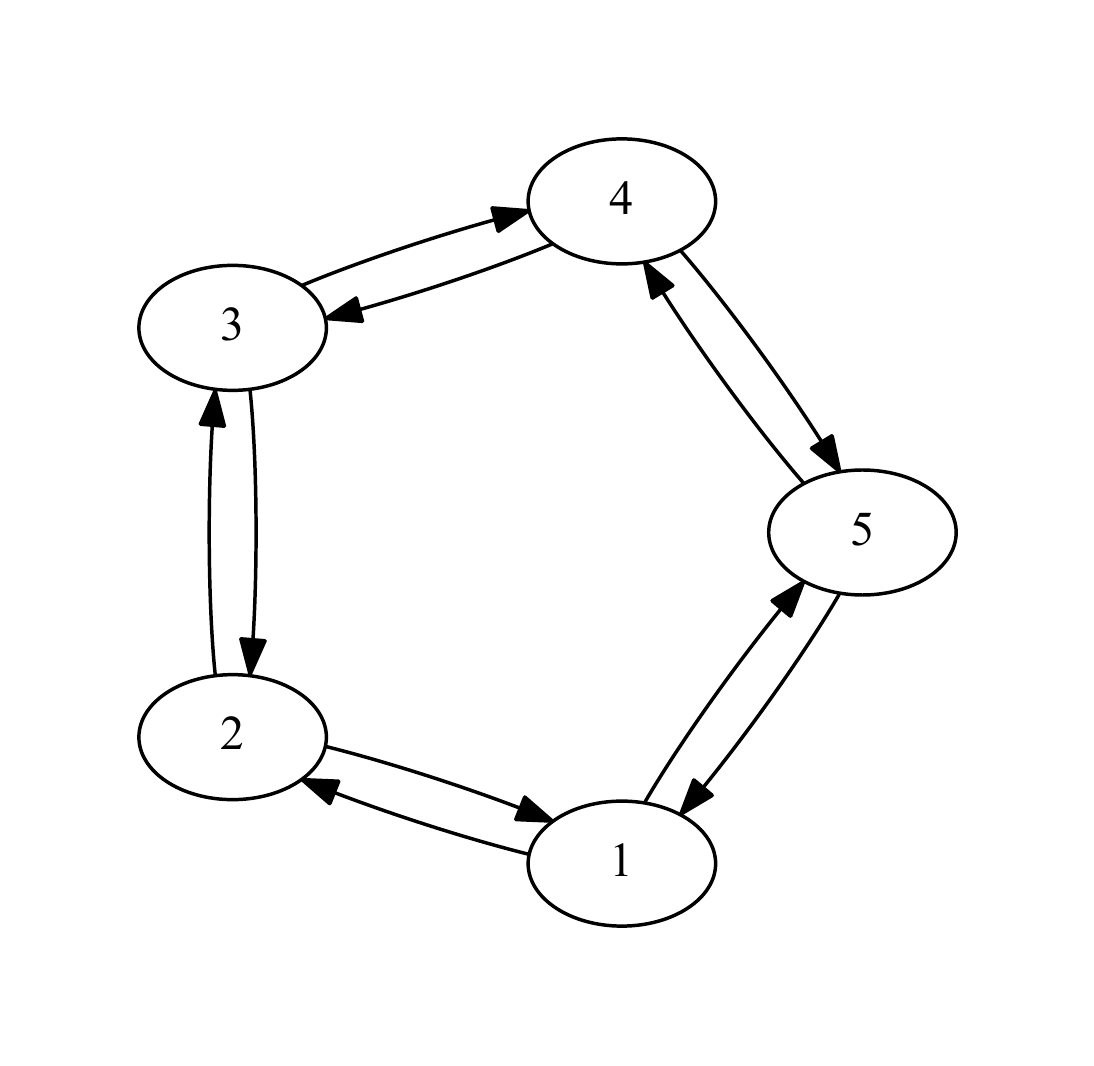} \label{fig:ring5}}
\caption{The complete (a, b, c), star (d, e, f), and ring (g, h, i) networks used to explore the impact of node and node state quantities on the consensus time of unconstrained gossip with proportional selection.}
\label{fig:test_networks}
\end{figure} 
Each table shows the 95\% confidence interval of the mean consensus time for all initial states at a particular distance from consensus at the specific noded/states configuration\footnote{This means we averaged the consensus time for all states $d$ units from consensus, as $d$ ranges from $1$ to the maximum observed distance.}. An increase in the number of nodes corresponds to moving down a column, from top to bottom. An increase in the number of node states corresponds to moving across a row, from left to right. A value of ``-'' indicates that there was no data for the corresponding network configuration. This partitioning is important because, as we will see, the further the network starts from consensus the larger the impact of our measured variables. 

Table \ref{tab:hyp1_1} displays the data for initial network states that start at a distance of 1 unit from consensus. 
\begin{table}
\caption{Theoretically determined 95\% confidence intervals for the mean consensus times at distance = 1.}
\label{tab:hyp1_1}
\begin{center}
\begin{tabular}{l | l | r r r r}
\toprule
 & & 2 states & 3 states & 4 states \\
\midrule
\multirow{3}{*}{$K_n$} 
& 3 nodes	&(5.5, 5.5)		&(5.5, 5.5)		&(5.5, 5.5)   \\
& 4 nodes	&(6.13, 6.13)	&(6.13, 6.13)	&(6.13, 6.13) \\
& 5 nodes	&(7.59, 7.59)	&(7.59, 7.59)	&(7.59, 7.59) \\
\midrule
\multirow{3}{*}{Star} 
& 3 nodes	& (4.79, 5.88) & (5.09, 5.57) & (5.18, 5.49)\\
& 4 nodes	& (7.11, 7.89) & (7.31, 7.69) & (7.38, 7.62) \\
& 5 nodes	& (10.01, 10.61) & (10.16, 10.46) & (10.21, 10.41) \\
\midrule
\multirow{3}{*}{Ring} 
& 3 nodes	&(5.50, 5.50) & (5.50, 5.50) & (5.50, 5.50) \\
& 3 nodes	&(5.89, 5.89) & (5.89, 5.89) & (5.89, 5.89) \\
& 3 nodes	&(7.31, 7.31) & (7.31, 7.31) & (7.31, 7.31) \\
\bottomrule
\end{tabular}
\end{center}
\end{table}
It can be observed that increasing the number of nodes in all networks results in a significant increase in the consensus time; however, increasing the number of states per node does not correspond to a significant increase in the consensus time. The response to increasing the number of node states is most likely occurs because at a distance of $1$, the network is already close to consensus; typically with only a single node out of sync. Because so few nodes require a state change, the number of states will have little to no impact on the time required to reach a consensus. 

Table \ref{tab:hyp1_2} displays the data for initial network states that start at a distance of 2 unit from consensus. 
\begin{table}
\caption{Theoretically determined 95\% confidence intervals for the mean consensus times at distance = 2}
\label{tab:hyp1_2}
\begin{center}
\begin{tabular}{l | l | r r r r}
\toprule
 & & 2 states & 3 states & 4 states \\
\midrule
\multirow{3}{*}{$K_n$} 
& 3 nodes	&-				&(7.33, 7.33)		&(7.33, 7.33) \\
& 4 nodes	&(7.73, 7.73)	&(8.34, 8.63)		&(8.57, 8.70) \\
& 5 nodes	&(10.56, 10.56)	&(10.9349, 11.1012)	&(11.13, 11.22)\\
\midrule
\multirow{3}{*}{Star} 
& 3 nodes	& - 			& (7.00, 7.00)   & (7.00, 7.00) \\
& 4 nodes	& (9.50, 9.50) 	& (10.24, 10.59) & (10.52, 10.68) \\
& 5 nodes	& (14.37, 14.47)& (14.94, 15.17) & (15.20, 15.33) \\
\midrule
\multirow{3}{*}{Ring} 
& 3 nodes	&- 				& (7.33, 7.33) 	& (7.33, 7.33) \\
& 4 nodes	&(7.20, 7.66) 	& (7.99, 8.27) 	& (8.21, 8.34) \\
& 5 nodes	&(10.00, 10.25) & (10.46, 10.63) & (10.64, 10.73) \\
\bottomrule
\end{tabular}
\end{center}
\end{table}	
It can be observed that increasing the number of nodes in all networks results in a significant increase in the consensus time; increasing the number of states per node results in a small, and questionably significant increase in the consensus time. The impact on the consensus time as a result of increasing the number of node states appears to grow in significance with the number of nodes. If the observed trends continue, then the impact should be clearly significant for large networks. This is in line with our expectations based on the growth of the Markov chain state space. 

Table \ref{tab:hyp1_3} displays the data for initial network states that start at a distance of 3 unit from consensus.
\begin{table}
\caption{Theoretically determined 95\% confidence intervals for the mean consensus times at distance = 3}
\label{tab:hyp1_3}
\begin{center}
\begin{tabular}{l | l | r r r r}
\toprule
 & & 2 states & 3 states & 4 states \\
\midrule
\multirow{3}{*}{$K_n$} 
& 3 nodes	&-	&-				&- \\
& 4 nodes	&-	&-				&(9.99, 9.99) \\
& 5 nodes	&-	&(12.70, 12.70)	&(13.03, 13.10)\\
\midrule
\multirow{3}{*}{Star} 
& 3 nodes	&-	&-				&- \\
& 4 nodes	&-	&-				&(12.25, 12.25) \\
& 5 nodes	&-	&(17.35, 17.36)	&(17.81, 17.91)\\
\midrule
\multirow{3}{*}{Ring} 
& 3 nodes	&- & - 				& - \\
& 3 nodes	&- & - 				& (9.54, 9.54) \\
& 3 nodes	&- & (12.07, 12.14) & (12.41, 12.48) \\
\bottomrule
\end{tabular}
\end{center}
\end{table}
As with tables \ref{tab:hyp1_1} and \ref{tab:hyp1_2} an increase in the number of nodes corresponds to a significant increase in the consensus time. Unlike tables \ref{tab:hyp1_1} and \ref{tab:hyp1_2}, however, an increase in the number of states per node also results in a significant increase in the consensus time. This observation supports our expectation that the distance between an initial state and consensus is important.

Overall, these observations paint an interesting picture of the dynamics behind the temporal behavior of unconstrained gossip. Across all three graphs that we examine, increasing the number of nodes in the network results in an increase in the consensus time. Increasing the number of states per node, however, has different effects depending on the distance of the initial network state. When the initial network state is close to consensus, the adding more possible states per node has little to no impact on the consensus time. As the distance between the initial network state and consensus increases, adding additional choices for each node states results in a more significant impact on the consensus time. 

In general, our results suggest that the consensus time increases with both the number of nodes in the network and the number of possible states per node. They also reflect our expectations of behavior as a result of Markov chain analysis. The Markov state space grows faster via the addition of nodes than it does via the addition of states per node. Under this perspective, one possible cause for the increase in consensus time is simply that it takes longer to reach a consensus because the state space is larger, and thus the likelihood of encountering the optimal consensus sequence decreases and the average path length to a consensus state increases.

\subsection{The Impact of Network Density on Consensus Time}

So far we have established that the consensus time of unconstrained gossip under proportional selection is largely influenced by the number of nodes in the underlying communication network. This finding is in line with the existing research on the average gossip algorithm and gossip algorithms that do not account for conflict among transmissions. We have also shown that the number of states per node can have a significant impact on the consensus time, but that the significance of that impact varies with the distance of the initial network state from consensus. To the best of our knowledge, similar results concerning the number of states per node have not been documented. 

We now investigate how the density\footnote{We use the standard definition of network density as the ratio of existing edges to total possible edges.} of a communication network influences the consensus time. Network density is worth investigating because it impacts how easily information can spread between nodes. If a network is too sparse, it may lack a directed spanning tree and thus be unable to produce a consensus. As a network grows denser, it is possible for multiple directed spanning trees to exist, thereby providing multiple options for consensus formation, with some consensus sequences requiring more time than others to complete. Regardless of the density, if the network is well connected (see theorem \ref{thm:consensus}) then there exists a shortest consensus sequence. The existence of this shortest consensus sequence is why we do not necessarily think that the density of a network matters; no matter how many edges get added to a network the consensus sequences they produce can be no better than the shortest and there is no guarantee that they will be used for transmission. In terms of our Markov modeling framework, the network density has a direct impact on the transition probabilities of the Markov chain. An increase in density has the potential to dilute the probability of transmitting along the shortest consensus sequence, but it does not nullify it. 

We conduct this investigation by constructing the completely connected five node network $K_{5}$ with $3$ possible node states and then removing edges at random until a desired density is reached. To account for the multiple configurations that can occur when edges are removed, we average the consensus time of $30$ graphs at each density value.

Figure \ref{fig:experiment6a} shows the results of our investigation as a plot of the consensus time as the density of an arbitrary network increases from $0.6$ to $1.0$ in increments of $0.05$. Each line represents a different initial distribution of node states. The bars on each line bound the 95\% confidence interval of the mean consensus time at the corresponding density value.
\begin{figure}
\centering
\includegraphics[width=80mm]{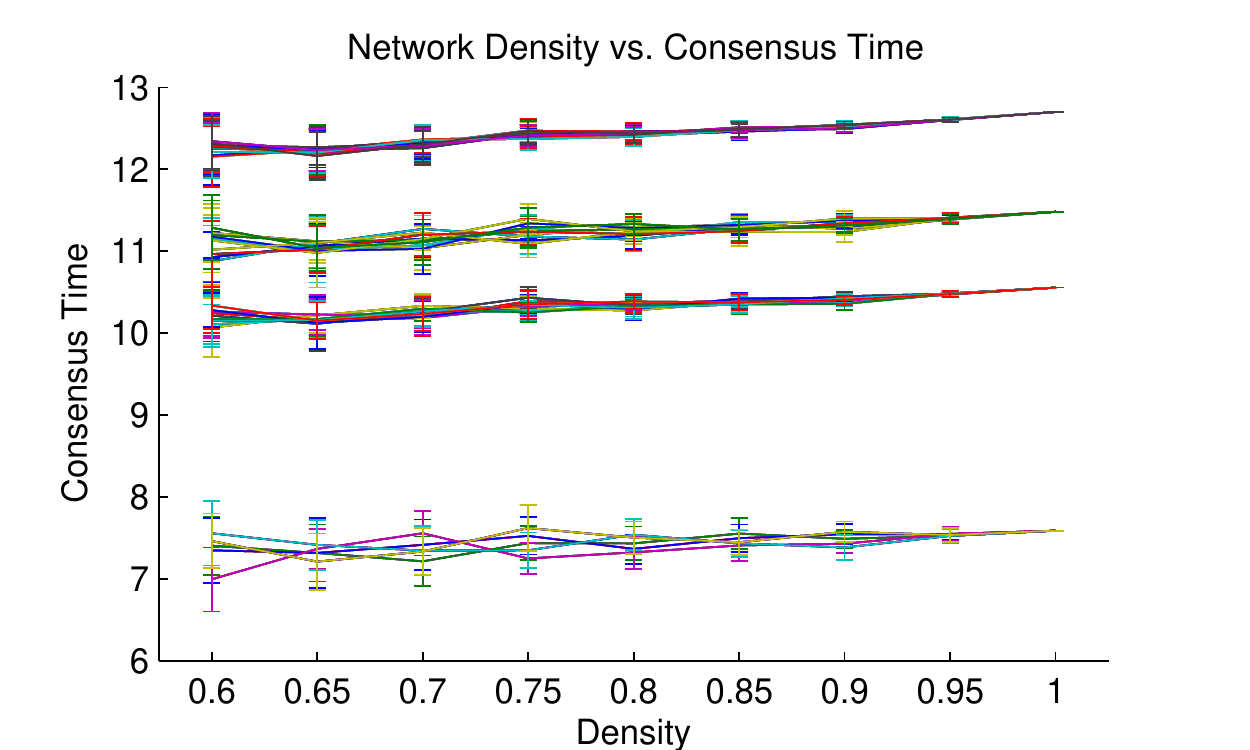}
\caption{The consensus time as the density of an arbitrary network with 5 nodes and 3 states per node increases from $0.6$ to $1.0$ in increments of $0.05$. Each line represents a different initial distribution of node states; the exact state of each line is irrelevant. Each bar represents the 95\% confidence interval of the associated data point.}
\label{fig:experiment6a}
\end{figure} 
There are two primary observations that can be made from this data. The first is that the initial states cluster together into partitions. This is the same behavior that was observed during our validation experiments; see table \ref{tab:experiment1a} and \ref{tab:experiment1b}. Not only can the initial network states be partitioned, but these partitions become clearer as the network becomes denser. At a density value of $0.6$ the boundaries between three of the four partitions are fuzzy, but once the density of the network reaches $0.7$ it is clear which initial states belong to each partition. 
The second observation is that regardless of the initial state, there does not appear to be a significant statistical increase in the consensus time as the density of the network increases, but the slight upward trend suggests that as the density increases the system may be more likely to perform at the slower end of the performance spectrum; as evidenced by the increase in the lower bound of the confidence intervals while the upper bounds remain fairly fixed. For instance, at a density of $0.6$ on the examined network an initial state one unit away from consensus it should take, on average, between $6.5$ and $8$ steps for a consensus to be formed; however, when the density increases to $0.95$ is should take an average of $7.5$ steps to form a consensus.

As previously mentioned, this observed behavior may be due to the additional directed spanning trees that appear as a network becomes more connected. As more viable paths to consensus become available it is feasible that the probability of moving along the shortest path decreases, and as a result longer transmission routes are more likely to be selected, thus resulting in the possibility (but not guarantee) of an increased time to consensus.

\section{Conclusions}

Gossip algorithms are widely used to solve the distributed consensus problem on networks, but issues can arise when nodes receive multiple signals either at the same time or before they are able to finish processing their current work load. In the real world, it can often be hard to limit the amount of information a node is exposed to, especially as networks become larger and their nodes more complex. To address these issues, we introduce the notion of conflict resolution for \emph{unconstrained} gossip algorithms and prove that their application leads to a valid consensus state when the underlying communication network possesses certain properties. We also introduce a methodology that is based on absorbing Markov chains for analyzing unconstrained gossip algorithms that makes use of these conflict resolution algorithms. This technique allows us to calculate both the probabilities of converging to a specific consensus state and the time such convergence is expected to take. Finally, we make use of simulation experiments to verify and supplement our theory with additional results.

We show that the number of nodes in a network, the initial state of the network, and specific network topology are all critical factors in determining how long it will take for consensus formation. The number of possible states that each node can assume has much less impact on system performance than a increase in the number of nodes. Furthermore, the significance of this impact varies with the initial distribution of node states. This finding has important implications for deriving bounds on consensus time for the behavior of unconstrained gossip algorithms \emph{without} first requiring computation of the Markov transition matrix. Specifically, it suggests that it might be possible to leverage existing techniques from the existing body of research on gossip algorithms, as well as techniques from the opinion dynamics and voter model literature in order to study unconstrained gossip when assumptions are made on the initial state of the network. We think that such a bridge is not only a good idea, but essential for future advancements in distributed problem solving.

\section*{Acknowledgments}

This work was supported in part by ONR grant \#N000140911043 and General Dynamics grant \#100005MC.

\bibliographystyle{IEEEtran}
\bibliography{references}

\end{document}